\documentclass[journal,12pt, onecolumn]{IEEEtran}
\usepackage{amsmath}
\usepackage{bm}
\usepackage{algorithmic}
\usepackage{hyperref}
\usepackage{algorithm}
\usepackage{subfigure}
\usepackage{epsfig}
\usepackage{cite}
\usepackage{color}
\usepackage{multirow,makecell}
\usepackage{amssymb}
\usepackage{graphicx}

\usepackage{setspace}

\newtheorem{theorem}{Theorem}
\newtheorem{lemma}{Lemma}
\newtheorem{proof}{Proof}

\usepackage{algorithm}
\def\BibTeX{{\rm B\kern-.05em{\sc i\kern-.025em b}\kern-.08em
		T\kern-.1667em\lower.7ex\hbox{E}\kern-.125emX}}

\begin{document}
	
	\title{Robust Secure Transmission for Active RIS Enabled Symbiotic Radio Multicast Communications
	 \footnote{This article was presented in part at the IEEE VTC2022-Fall, Beijing, China, September 26-29, 2022 \cite{Conference}. }
}

\author{Bin Lyu,~\IEEEmembership{Member,~IEEE,}
 Chao Zhou,
 Shimin Gong,~\IEEEmembership{Member,~IEEE,}\\
 Dinh Thai Hoang,~\IEEEmembership{Senior Member,~IEEE,}
 and~Ying-Chang Liang, ~\IEEEmembership{Fellow,~IEEE}
\IEEEcompsocitemizethanks{\IEEEcompsocthanksitem B. Lyu and C. Zhou are with the Key Laboratory of Ministry of Education in Broadband Wireless Communication and Sensor Network Technology, Nanjing University of Posts and Telecommunications, Nanjing 210003, China (email: blyu@njupt.edu.cn, zoe961992059@163.com). 
 S. Gong is with School of Intelligent Systems Engineering, Sun Yat-sen University, Guangzhou 510275, China (email: gongshm5@mail.sysu.edu.cn).
 D. T. Hoang is with School of Electrical and Data Engineering, University of Technology Sydney, Sydney, NSW 2007, Australia (email: hoang.dinh@uts.edu.au). 
 Y.-C. Liang is with the Institute for Infocomm Research, A*STAR, Singapore (e-mail: liangyc@ieee.org).
}}
	\maketitle
	\begin{abstract}
		In this paper, we propose a robust secure transmission scheme for an active  reconfigurable intelligent surface (RIS) enabled symbiotic radio (SR)  system in the presence of multiple eavesdroppers (Eves). In the considered system, the active RIS is adopted to enable   
		the  secure transmission of primary signals   from the primary transmitter to multiple primary users in a multicasting manner, and simultaneously achieve its own information delivery to the secondary user by riding over the primary signals. Taking into account the imperfect channel state information (CSI) related with Eves, we formulate the system power consumption minimization problem by optimizing the transmit beamforming  and reflection beamforming  for  the bounded and statistical CSI error models, taking  the worst-case SNR constraints and the SNR outage probability constraints at the Eves into considerations, respectively. Specifically, the S-Procedure and the Bernstein-Type Inequality are implemented to approximately transform the worst-case SNR and the SNR outage probability constraints into tractable forms, respectively.
		 After that, the formulated problems can be solved by the proposed alternating optimization (AO) algorithm with the semi-definite relaxation and sequential rank-one constraint relaxation techniques. Numerical results  show that the proposed active RIS scheme can reduce up to 27.0\% system power consumption  compared to the passive RIS.
	\end{abstract}
	
	\section{Introduction}
For the forthcoming 6G communications, the number of wireless devices is expected to be numerous for ubiquitous connections. However, in order to achieve this goal, there still exist some gaps needed to be addressed. On the one hand, it is extremely hard to meet 
 the requirement of spectrum resources for supporting the  ubiquitous connections as the required spectrum resources can be  nearly 76 GHz \cite{Spectrum}. 
 In the past decades, 
cognitive radio (CR) is considered to be a promising way to address the scarcity of spectrum resources, in which the spectrum bands assigned to the primary transmission can be opportunistically exploited by the secondary transmission. However, the secondary transmission could interfere the primary system inevitably, which fundamentally limits the practical implementation of CR \cite{LiangSurvey}. On the other hand, the spectrum sensing  functionality and active radio frequency (RF) components used in CR networks  result in high power consumption,  which disobeys the energy efficiency requirement for 6G communications. Thus, innovative technologies for achieving spectrum- and energy- efficient communications are in urgent needs.

Symbiotic radio (SR) has been proposed as a promising technology for 6G communications, which aims at constructing a  mutualism relationship between the primary and secondary transmissions under the quality of service (QoS) requirements of spectrum and energy efficiency \cite{LongIoT,GuoWCL,Chu,YongZeng}. In a typical SR system, secondary users (SUs) deliver their own information by passively backscattering the primary signals via periodically adjusting their load impedance \cite{LongIoT}. As such, the exclusive spectrum band for the secondary transmission is not necessary. Moreover, due to the passive backscattering characteristic, active RF components and dedicated incident signal sources can be avoided for the SUs, which is a key factor for obtaining a  satisfying energy efficiency. As a reward for  the primary transmission, the secondary transmission provides additional multi-path for enhancing its performance \cite{LiangSurvey}. In \cite{GuoWCL}, the SR  paradigm extended from backscatter communication was first proposed, and the symbiotic relationships between the primary and secondary transmissions, i.e., parasitic and commensal, were introduced. 
In \cite{LongIoT}, the  weighted sum-rate maximization problems were investigated for multiple-input-single-output (MISO) systems under both parasitic and commensal SR setups. In \cite{Chu}, the finite blocklength channel codes were used for  the secondary transmission in SR systems. In \cite{YongZeng}, a channel estimation method for cell-free SR systems was proposed, where the state of the SU is switched for  estimating the channel state information (CSI) of the direct links and the cascaded reflecting links. In \cite{LongIoT,GuoWCL,Chu,YongZeng}, backscatter devices each equipped with single antenna were served as the SUs for the secondary transmission, which limits the performance of the SR system from the following two aspects. Firstly, compared to the direct link for the primary transmission, the backscatter link suffers from the double-fading effect  such that the communication efficiency for the secondary transmission is unsatisfying.  Secondly, since the additional multi-path provided by the backscatter link  is quite weak, the  performance improvement for the primary transmission is insignificant.

Recently, reconfigurable intelligent surface (RIS)  has been widely used for performance enhancement of wireless communications \cite{Huang2019IRS,Gong,Wu2020Survey,PanSurvey}. An RIS is generally  equipped with massive reflecting elements, the amplitude reflection coefficients and phase shifts of which can be collaboratively and smartly adjusted for information reflection. According to the amplification characteristic, the RIS can be classified into three categories of passive RIS, hybrid RIS and active RIS \cite{Dai}. Specifically, the passive RIS can only reflect the incident signal without any amplification, i.e., the maximum amplitude reflection coefficient is one. In contrast, the active RIS can amplify the incident signals by consuming additional power \cite{ActiveRISLong}. While, the hybrid RIS is a combination of passive RIS and active RIS.
There exist a lot of works integrating the passive RIS into SR systems for enhancing the efficiency of both the primary and secondary transmissions   \cite{GangYang,QianqianZhang,HuaSR,HuaSRTwo,Broadcasting}.  In \cite{GangYang}, the passive  RIS was operated as a helper for assisting the information transmission from a primary transmitter (PT) to a primary user (PU) and vice versa. Alternatively,  the successive interference cancellation (SIC) and maximal ratio combining (MRC) techniques were used for signal decoding. Different from \cite{GangYang} in which the passive RIS was introduced into the existing SR systems, the passive RIS can also serve as an SU \cite{QianqianZhang,Broadcasting,HuaSR,HuaSRTwo}. Specifically, the passive RIS can not only act as a helper for the primary transmission, but also modulate the primary signal for delivering its own information \cite{QianqianZhang}. In \cite{HuaSR}, an unmanned aerial
vehicle (UAV) was leveraged in RIS assisted SR systems, for which the  UAV trajectory is optimized for performance enhancement. In \cite{HuaSRTwo}, the bit error rate minimization problem was  studied under the constraints about primary transmission rates for  both
parasitic and commensal SR setups. In \cite{Broadcasting}, a broadcasting communication scheme was proposed for the scenario with multiple PUs. In the prior works  \cite{GangYang,QianqianZhang,HuaSR,HuaSRTwo,Broadcasting}, since the passive RIS cannot amplify the incident signals, it is impossible  to eliminate the double-fading effect fundamentally, which  results in the  limited performance improvement. To break this fundamental limitation, the active RIS with amplification capability has been investigated in the literature \cite{ActiveRISLong,YouActive,PanActive}. In \cite{ActiveRISLong}, the active RIS was  proposed, where each reflecting element equipped with active loads can reflect and amplify the incident signal. Specifically, an active RIS empowered single input multiple output (SIMO) communication system was investigated to boost the system spectrum and energy efficiency.
In \cite{YouActive},  the active RIS aided downlink and uplink communications were studied and the deployment insight of the active RIS was revealed. In \cite{PanActive}, the performance comparison between the passive RIS and active RIS under the same power budget was investigated, the results confirmed the performance superiority of using the active RIS. Inspired by this observation, it is nature to apply the active RIS to the SR systems to achieve satisfying system performance. Specifically, by exploiting  the amplification characteristic, the primary signal power reflected by the active RIS can be amplified,  thus improving the primary transmission efficiency. At the same time, the secondary transmission by riding over the primary signal can also be enhanced. More importantly, compared to the passive RIS, the additionally consumed power for amplification at the active RIS is limited \cite{PanActive,Amato}. Hence, it is worth investigating the active RIS enabled SR systems, which has not been considered in the literature yet. 

Furthermore, it  is well known that wireless medium is open and broadcasting, the useful information delivery in the SR systems may be  vulnerable to potential eavesdroppers (Eves). However, it may not be suitable for encrypting the primary signal to enable the secondary transmission, i.e., guaranteeing the successful decoding of the secondary signal at the SU. To ensure the security of wireless transmission, physical layer security (PLS) technique has been considered as  a promising solution, which utilizes the wireless channels' inherent characteristic to reduce the achievable signal-to-noise-ratio (SNR) at the  Eves. It is worth noting that the PLS technique can be naturally implemented in RIS assisted communication systems and has been widely studied \cite{Chu2020Secure,Wang,ActiveSecure,PanTwo,GeSecure}. In \cite{Chu2020Secure}, the passive RIS was used to guarantee the secure transmission for traditional multi-antenna communication systems under the existence of an Eve. In \cite{Wang}, the secure transmission performance was further improved by using the active RIS. In \cite{ActiveSecure}, the secure broadcasting scheme was proposed for passive RIS aided SR systems and the secrecy  rate was maximized by optimizing the hybrid precoder at the PT and reflection coefficients at the RIS.
However, the CSI was assumed to be perfectly acquired in \cite{Chu2020Secure,Wang,ActiveSecure}, which is very difficult in practice due to the fact that the Eves tend not to cooperate with the PT for channel estimation. Considering the imperfect CSI, the authors  in \cite{PanTwo} investigated the robust secure transmission design in passive RIS assisted MISO communication systems. In \cite{GeSecure}, the passive RIS enabled two-way secure transmission scheme was proposed under the bounded CSI error model. However, the study of secure transmission for  SR systems is still at its early stage, and  no works  have studied the robust secure transmission scheme for SR systems under the imperfect CSI scenarios, which motivates the study of this paper.

In this paper, we propose an active RIS enabled robust secure transmission scheme for  SR  systems in the presence of multiple Eves, which consists of one PT, one RIS, one SU and multiple PUs. Specifically, the RIS is deployed to assist the confidential information multicasting from the PT to the PUs against the interception from the Eves. In the meanwhile, the RIS modulates the primary signal using the binary phase shift keying (BPSK) technique to enable its own information delivery to the SU. 
Considering the fact that the legitimate users are willing to cooperate with the PT, the CSI related with the PUs and SU is considered to be   perfectly known.
Taking into account the difficulty of acquiring the CSI related with the Eves, we investigate the robust beamforming designs for both the bounded CSI error model and the statistical CSI error model.
% For each CSI error model, we aim to minimize the total system power consumption under  the achievable SNR constraints at  the SU, PU and Eves. 
Compared to \cite{GangYang,QianqianZhang,HuaSR,HuaSRTwo,Broadcasting}, this work focuses on the robust secure transmission design for the SR systems, which is much more challenging due to the  secure transmission constraints with the imperfect CSI models for the Eves. Furthermore,
in contrast to \cite{ActiveSecure}, our proposed secure transmission scheme is  robust to address the CSI errors and can avoid the design mismatch caused by the perfect CSI assumptions for the Eves.  It is noted that the secure transmission schemes  in \cite{PanTwo,GeSecure} are not applicable  due to the existence of both the primary transmission and secondary transmission. The main contributions of this work are summarized as follows:
 \begin{itemize}
\item{To the best of our knowledge, this work is the first to investigate the robust secure transmission scheme for   RIS enabled SR systems by taking into account the CSI errors related with the Eves. In particular, the active RIS is adopted to not only establish  the robust secure transmission from the PT to the PUs against the  presence of Eves, but also achieve the secondary information delivery from itself to the SU. Moreover, for both the bounded CSI error model and the statistical CSI error model, the robust secure transmission designs  are proposed, respectively.}
\item{For the bounded CSI error model, we aim to minimize the total system power consumption  by jointly optimizing the transmit beamforming at  the PT and reflection beamforming at the active RIS under the worst-case SNR constraints with imperfect CSI at the Eves and minimum achievable SNR constraints at the PUs and SU. To address the non-convexity of the formulated problem, we propose an alternating optimization (AO) algorithm to optimize the transmit beamforming and reflection beamforming in an alternating manner, for which the S-Procedure is implemented to approximately transform the worst-case SNR constraints  into  equivalent linear matrix inequalities (LMIs). For the optimization of the transmit beamforming, we utilize  the semi-definite relaxation (SDR) to relax the rank-one constraint and  prove its tightness for obtaining the optimal solution. While, the sequential rank-one constraint relaxation (SROCR) technique is adopted to obtain a sub-optimal solution to the reflection beamforming.}
\item{For the statistical CSI error model, we formulate a system power consumption minimization problem subject to the SNR outage probability constraints at the Eves and minimum achievable SNR constraints at the PUs and SU.  Specifically, the SNR outage probability constraints  are properly approximated  into  tractable forms by using the Bernstein-Type Inequality. In addition, the AO algorithm with the SDR and SROCR techniques is also proposed to solve the formulated problem.}

 \item{  Numerical results are demonstrated to evaluate the performance of the proposed schemes. To show the effectiveness  of using the active RIS, the robust secure transmission design for the passive/hybrid RIS enabled SR systems is discussed. Numerical results confirm that compared to the passive RIS, the active RIS can  reduce up to 27.0\% system power consumption. In addition, the proposed scheme under the statistical CSI error model always consumes less power than that under the bounded CSI error model.}

 \end{itemize}

\begin{figure}
	\centering
	\includegraphics[width=0.7 \linewidth]{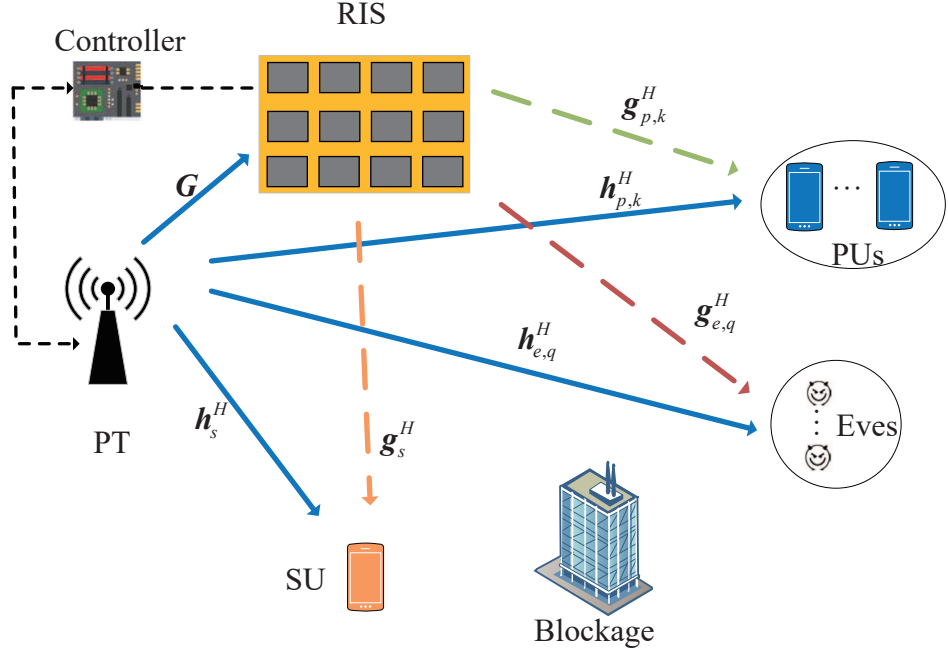}
	\caption{An RIS enabled secure SR system.}
	\label{Illustration}
\end{figure}

The rest of this paper is organized as follows. Section \ref{sysmod} describes the system model. In Section  \ref{BoundedCSI} and Section \ref{StatisticalCSI}, the robust beamforming designs for  the bounded and statistical CSI error models are investigated, respectively.  Section  \ref{NumericalResults} provides numerical results for performance verification. The conclusion of this paper is summarized in Section \ref{Conculsions}.

	\section{System Model}
	\label{sysmod}
	As shown in Fig. \ref{Illustration},  we consider the secure communication of an active RIS enabled SR multicast system\footnote{The system model can  be degenerated to the scenario with a passive RIS, and the proposed scheme in this paper is still applicable, which will be briefly discussed in Section \ref{SectionBechmark}.},
	 which consists of  one PT equipped with $N$ ($N \ge 1$) antennas, one active RIS equipped with $M$ ($M \ge 1$) reflecting elements, $K$ ($K \ge 1$) single-antenna PUs, one single-antenna SU, and $Q$ ($Q \ge 1$) single-antenna Eves. There exist the primary transmission from the PT to the PUs and the secondary transmission from the RIS to the SU. Specifically,  the PT transmits primary signals containing confidential information  to all the PUs in  a  multicast communication manner. In the system under considerations, the Eves are assumed to be located near the PUs to intercept the information transmitted to them. In this context, the RIS is deployed to establish reliable links from the PT to the PUs, and at the same time deliver its own information to the SU by riding over the primary signals. We consider that the links between the Eves and the SU are unavailable due to the blockages between them (as the main goal of the Eves is to intercept the information transmitted to the PUs), thus the received information at the SU cannot be intercepted by the Eves. We denote the set of antennas at the PT, the set of RIS reflecting elements, the set of PUs, and  the set of Eves as $\mathcal{N} = \{1,\ldots,N \}$, $\mathcal{M} = \{1,\ldots,M\}$,  $\mathcal{K} = \{1,\ldots,K\}$, and $\mathcal{Q} = \{1,\ldots,Q\}$,  respectively.

	We consider a block flat-fading channel, i.e., the CSI is invariant during one secondary symbol period \cite{QianqianZhang}. The complex baseband equivalent channel from the PT to the RIS, from the PT to the $k$-th PU, from the PT to the SU, from the PT to the $q$-th Eve, from the RIS to the $k$-th PU, from the RIS to the $q$-th Eve, and from the RIS to the SU are denoted by  ${\bm{G}} \in {{\mathcal{C}}^{M \times {N}}}$, $\bm{h}_{p,k}^H \in \mathcal{C}^{1\times N }$, $\bm{h}_s^H \in \mathcal{C}^{1 \times N}$, $\bm{h}_{e,q}^H \in \mathcal{C}^{1 \times N}$, ${\bm{g}}_{p,k}^H \in {{\mathcal{C}}^{1 \times M}}$, ${{\bm{g}}_{e,q}^H} \in {{\mathcal{C}}^{1 \times M}}$, and ${{\bm{g}}_{s}^H} \in {{\mathcal{C}}^{1 \times M}}$,  respectively, where $k \in \mathcal{K} $, and $q \in \mathcal{Q}$.

\subsection{Signal Transmission Model}	
We consider the commensal SR setup, i.e., the symbol rate of the primary transmission  is much larger than that of the secondary transmission \cite{LongIoT}. In particular, $L$  ($L \gg 1$) primary symbols can be transmitted by the PT during one secondary symbol period. Denote the $l$-th symbol of the primary and secondary transmissions  as $s(l)$ and $c$, respectively, where $s(l) \sim \mathcal{CN} (0,1)$, $l \in \mathcal{L}$, and $\mathcal{L} = \{1,\ldots,L \}$ represents the set of all primary symbols.  We assume that the RIS modulates its own information over the primary information by the  BPSK scheme, i.e., $c \in \{-1,1\}$. It is worth noting that the BPSK modulation scheme  is practical and widely used for the RIS due to its reflection characteristic \cite{QianqianZhang,Broadcasting}, which can be achieved by adjusting the phase shifts.

We denote the amplitude reflection coefficient and  phase shift of the $m$-th reflecting element at the RIS as $\beta_{m}$ and  $\theta_{m}$, respectively, where $\beta_m \in [0,\eta]$, $\theta_m \in [0,2\pi)$, and $ \eta $ is the  maximum amplification factor of the reflecting elements. The reflection coefficient for the $m$-th reflecting element is thus expressed as $c v_m$ with $v_m = \beta_m e^{j\theta_{m}}$, where $m\in \mathcal{M}$.
At the cost of additional power consumption, the active RIS can amplify the incident signals and the maximum amplification factor is generally not smaller than one, i.e.,   $ \eta \ge  1 $  \cite{Dai,ActiveRISLong}.

During one  secondary symbol period of interest, the $l$-th signal received by the $k$-th PU can be expressed as
	\begin{equation}
		\begin{aligned}	
		\label{signalpk}
		y_{p,k}(l) &= \bm{h}_{p,k}^H \bm{w} s(l) + \bm{g}_{p,k}^H (c \bm{\Phi})   \bm{G} \bm{ w} s(l) \\
		 &+\bm{g}_{p,k}^H (c \bm{\Phi})  {\bm{n}_I(l) } + n_{p,k}(l), ~l \in \mathcal{L}, k \in \mathcal{K},
		\end{aligned}
	\end{equation}
where $\bm{w} \in \mathcal{C}^{N\times 1}$ is the transmit beamforming vector at the PT, $\bm \Phi \in {{\bf{C}}^{M \times {M}}} $ is the diagonal matrix with $v_m$ being its $m$-th diagonal element, ${\bm{n}}_I(l) \sim \mathcal{CN}(0,\sigma_{I}^2{\bm{I}_M}) $ is the thermal noise generated at the RIS due to the existence of amplifiers{\footnote{ Due to the amplify characteristic of the active RIS, the thermal noise cannot be ignored \cite{Dai,ActiveRISLong}.}, and $ {n_{p,k}(l) \sim \mathcal{CN}(0,\sigma_{p,k}^2)} $ is the additional Gaussian white noise (AWGN) at the $k$-th PU. Since the duration of the secondary symbol is much longer than that of the primary symbol, $ \bm g_{p,k}^H ( c\bm{\Phi} ) \bm G $ contained in \eqref{signalpk} can be considered as a slowly varying channel\cite{LongIoT}. Thus, we can consider that $ s(l) $  is transmitted from the PT to  the $k$-th PU through the equivalent channel $\bm h_{p,k}^H+ \bm g_{p,k}^H ( c\bm{\Phi } )\bm {G} $. Then, 
with a fixed $c$,	the instantaneous SNR for  decoding $ s(l) $ at the $k$-th PU can be expressed as  
	$ \bar{\gamma}_{p,k} (c) = \frac{{\left|  \left( \bm h_{p,k}^H + \bm{g}_{p,k}^H ( c\bm{\Phi } ) \bm G \right) \bm w \right|}^2 } {\sigma_{p,k}^2 + \sigma_I^2 || \bm g_{p,k}^H \bm{ \Phi } c  ||_2^2 }. $
 Since the  BPSK modulation scheme is applied at the RIS, the average SNR for the $k$-th PU  to decode $ s(l) $ is given by 
%	\begin{equation}
		\begin{align}
			\label{SNRPU}
		&{\gamma_{p,k} } = {E_c}  [ \bar{\gamma}_{p,k} (c)  ] \nonumber \\
		&= \frac{{\left|  \left( \bm h_{p,k}^H + \bm g_{p,k}^H {\bm \Phi } \bm G \right)\bm w \right|}^2  }{2\left(   {\sigma_{p,k}^2 +\sigma_I^2 || \bm g_{p,k}^H { \bm \Phi }  ||_2^2 }\right) }+\frac{{\left|  \left( \bm h_{p,k}^H - \bm g_{p,k}^H{\bm \Phi } \bm G \right)\bm w \right|}^2  }{2\left(   {\sigma_{p,k}^2 +\sigma_I^2 || \bm g_{p,k}^H { \bm \Phi }  ||_2^2 }\right)},~k \in \mathcal{K}.
		\end{align}
%	\end{equation}

	Similarly, the $l$-th signal received by the $q$-th Eve  can be expressed as
	\begin{equation}\label{signaleq}
		\begin{aligned}		
		{y_{e,q}}(l) &= \bm h_{e,q}^H\bm ws(l) + \bm g_{e,q}^H\left( {c\bm \Phi } \right)\bm G \bm ws(l) \\
		&+\bm g_{e,q}^H ( c \bm{\Phi} ) {{\bm{n}}_I(l)}+ n_{e,q}(l),~l \in \mathcal{L}, q \in \mathcal{Q},
		\end{aligned}
	\end{equation}
	where $ {n_{e,q}(l)\sim \mathcal{CN}(0,\sigma_{e,q}^2)} $ is the AWGN at the $q$-th Eve. The Eves aim at intercepting the primary information, i.e., $s(l)$. Thus, the average SNR for the $q$-th Eve to  decode $ s(l) $ can be written as
%	\begin{equation}
	\label{SNREQ}
		\begin{align}
			&{\gamma_{e,q}} = {E_c}\left[ \frac{{\left|  \left( \bm h_{e,q}^H + \bm g_{e,q}^H\left( {c\bm \Phi } \right)\bm G \right) \bm w \right|}^2 }{\sigma_{e,q}^2 +\sigma_I^2 || \bm g_{e,q}^H { \bm \Phi } c ||_2^2 }  \right   ] \nonumber \\
			&= \frac{{\left|  \left( \bm h_{e,q}^H + \bm g_{e,q}^H {\bm \Phi } \bm G \right)\bm w \right|}^2  }{2\left(   {\sigma_{e,q}^2 +\sigma_I^2 || \bm g_{e,q}^H { \bm \Phi }  ||_2^2 }\right) }+\frac{{\left|  \left( \bm h_{e,q}^H - \bm g_{e,q}^H{\bm \Phi } \bm G \right)\bm w \right|}^2  }{2\left(   {\sigma_{e,q}^2 + \sigma_I^2 || \bm g_{e,q}^H { \bm \Phi } ||_2^2 }\right)},~q \in \mathcal{Q}.
		\end{align}
%	\end{equation}

	We denote the received signal at the SU as $ {y_{s}}(l) $, which is formulated as 
	\begin{equation}\label{signals}
		\begin{aligned}
		{y_{s}}(l) &= \bm h_{s}^H\bm ws(l) + \bm g_{s}^H\left( {c\bm \Phi } \right)\bm G \bm ws(l) \\
		&+\bm g_{s}^H (c \bm{\Phi} ) {{\bm{n}}_I(l)} + n_{s}(l), ~l \in \mathcal{L}, 
		\end{aligned}
	\end{equation}
	where $ {n_{s}(l) \sim \mathcal{CN}(0,\sigma_{s}^2)} $ is the AWGN at the SU. According to\cite{QianqianZhang}, the maximum likehood detection technique can be adopted to decode $s(l)$ and $c$ jointly. Similarly, the average SNR for decoding $ s(l) $  at the SU, denoted by  ${\gamma_{s}^{(s)}} $, is given by
	\begin{equation}
	\label{SNRSS}
		\begin{aligned}
			{\gamma_{s}^{(s)}} = \frac{{\left|  \left( \bm h_{s}^H + \bm g_{s}^H {\bm \Phi } \bm G \right)\bm w \right|}^2  }{2\left(   {\sigma_{s}^2 +\sigma_I^2 || \bm g_{s}^H { \bm \Phi } ||_2^2 }\right) }+\frac{{\left|  \left( \bm h_{s}^H - \bm g_{s}^H{\bm \Phi } \bm G \right)\bm w \right|}^2  }{2\left(   {\sigma_{s}^2 +\sigma_I^2 || \bm g_{s}^H { \bm \Phi }  ||_2^2 }\right)}.
		\end{aligned}
	\end{equation}
	After decoding $s(l)$ at the SU, the SIC technique is applied to remove the first term of \eqref{signals}, i.e., $ \bm h_{s}^H\bm ws(l) $. After that, 
 the intermediate signal received by the SU in the considered secondary symbol period, denoted by $\hat{\bm{y}}_s=[\hat{y}_s(1),\ldots,\hat{y}_s(L)]^T$, can be formulated as 
 \begin{align}
 	&[\hat{y}_s(1),\ldots,\hat{y}_s(L)]^T =  \bm{g}_{s}^H \bm{\Phi} \bm{G} \bm{w}  [s(1),\ldots,s(L)]^T c  \nonumber \\
 	&+  [\bm g_{s}^H { \bm \Phi } {\bm{n}}_I(1), \ldots,\bm g_{s}^H { \bm \Phi }{\bm{n}}_I(L) ]^T c +  [n_{s}(1),\ldots,n_{s}(L)]^T. \nonumber
 \end{align}

By using the MRC technique, the SNR for decoding $c$ at the SU is given by 
	\begin{equation}
	\label{SNRSC}
		\begin{aligned}
			\gamma_{s}^{(c)} &=\frac{\sum \limits_{l=1}^L{\left|  \bm g_{s}^H {\bm \Phi } \bm G \bm w s(l)\right| }^2  }{\sigma_s^2 +\sigma_I^2 || \bm g_{s}^H { \bm \Phi } ||_2^2}   \overset{(a)}{\approx}\frac{L{\left|  \bm g_{s}^H {\bm \Phi } \bm G \bm w \right| }^2}{\sigma_s^2+\sigma_I^2 || \bm g_{s}^H { \bm \Phi } ||_2^2},  
		\end{aligned}
	\end{equation}
	where $(a)$	holds due to the fact that the arithmetic mean and the statistical expectation are approximately  equal if $L \gg 1$ \cite{QianqianZhang,Broadcasting}.

% It is worth noting that  by setting $\sigma_I^2 = 0$ and $\eta=1$ for the passive RIS scenario,   the corresponding SNRs of decoding  $s(l)$ and $c$ at different nodes  can be straightforward to obtain  based on \eqref{SNRPU}, \eqref{SNREQ}, \eqref{SNRSS}, and \eqref{SNRSC}, respectively. 

\subsection{System Power Consumption}
For the system with the active RIS, the total power consumption, denoted by $P_\text{tol}^\text{active}$, is expressed as 
\begin{align}
\label{TotalPower}
P_\text{tol}^\text{active}=P_\text{PT}+P_\text{RIS}^\text{active},
\end{align}
where $P_\text{PT} = ||\bm{w} ||^2$ is the transmit power of the PT, and $P_\text{RIS}^\text{active}$ denotes the power consumption of the active RIS. {According to \cite{ActiveRISLong}, for the active RIS with $M$ reflecting elements, its circuit power consumption can be estimated by
\begin{align}
\label{PowerRIS}
P_\text{RIS}^\text{active}=M P_\text{SW}+ M P_\text{DC} +P_\text{out}^\text{active},
\end{align}
where $P_\text{SW}$ is the power consumed by switching and controlling each reflecting element, $P_\text{DC}$ is  the direct current biasing
power at each reflecting element, and $P_\text{out}^\text{active}$ represents the output power at the active RIS and is formulated as  \cite{ActiveRISLong}} 
\begin{align}
\label{RISOutput}
{P_\text{out}^\text{active} = ||\bm{\Phi} \bm{G}\bm{w}  ||_2^2 + \sigma_I^2 || \bm{\Phi} ||_F^2.}
\end{align}
{From \eqref{TotalPower}-\eqref{RISOutput}, it is observed that the system power consumption is  controlled by the design of  transmit beamforming and reflection beamforming.}

\subsection{CSI Uncertainty Model}
In the literature, there are many  channel estimation methods proposed for RIS-assisted communication systems \cite{ChannelEstimation,ChannelEstimationCui,ChangshengOne,ChangshengTwo}. Specifically, the PUs and SU are willing to cooperate with the PT directly or through the RIS for obtaining the downlink CSI. Hence, it is reasonable to assume that the perfect CSI with respect to the direct and cascaded links of the PUs and SU is available at the PT, which is a common assumption for achieving the upper bound on system performance \cite{GangYang,QianqianZhang,HuaSR,HuaSRTwo,Broadcasting,Derrick}. 
 However, as the Eves tend not to cooperate with the PT for channel estimation. It is  challenging for acquiring the CSI related with the Eves.\footnote{For the scenario that the Eves are part of a multi-group primary communication system, the Eves and PUs can communicate with the PT in alternate secondary symbol periods. However, the Eves have no permission to access the information transmitted to the PUs and aim to intercept the information transmitted to the PUs when the PUs are being served. When the Eves are being served, their CSI can be acquired approximately.}  To take into account this issue, we adopt two widely-used error models, i.e., the bounded CSI error model and the statistical CSI error model, to  capture the  CSI uncertainty of $\bm{h}_{e,q}^H $ and ${\bm{g}}_{e,q}^H$, respectively.

\subsubsection{Bounded CSI Error Model} In this model, the maximum norm of all CSI estimation errors is bounded by an uncertainty radius, which is adopted to measure the level of CSI uncertainty. Specifically,
 $\bm{h}_{e,q}^H $ and  ${\bm{g}}_{e,q}^H$ can be respectively expressed as 
\begin{align}
\label{CSIheq}
	&\bm{h}_{e,q}=\hat {\bm{h}}_{e,q}+ \Delta \bm{h}_{e,q},
	~\Delta \bm{h}_{e,q}  \in {\mathcal{C}^{N \times {1}}}, \nonumber \\ 
	&~~~~~~~~{|| \Delta \bm{h}_{e,q}||_2 } \le \xi_{h,q}, ~\forall q \in \mathcal{Q},\\
%\end{align}
%\begin{align}
\label{CSIgeq}
		&\bm{g}_{e,q}=\hat {\bm{g}}_{e,q}+\Delta \bm{g}_{e,q},
		~\Delta \bm{g}_{e,q}  \in {{\mathcal{C}}^{M \times {1}}}, \nonumber \\
		&~~~~~~~~ {|| \Delta \bm{g}_{e,q}||_2 } \le \xi_{g,q}, ~\forall q \in \mathcal{Q}, 
\end{align}
 where $ \hat{\bm{h}}_{e,q} $ and $ \hat{\bm{g}}_{e,q} $ are  the estimated CSI  of $\bm{h}_{e,q}$ and $\bm{g}_{e,q}$ known at the PT, respectively, $ \Delta \bm{h}_{e,q} $ and  $ \Delta \bm{g}_{e,q} $ are the corresponding CSI estimation errors,  $ \xi_{h,q} $ and $ \xi_{g,q} $  are the  uncertainty radii  known at the PT.

\subsubsection{Statistical CSI Error Model} In this model,
	 {the CSI estimation errors of $\bm{h}_{e,q}^H $ and  ${\bm{g}}_{e,q}^H$ are assumed to follow the circularly symmetric complex Gaussian (CSCG) distribution \cite{PanTwo,Pan}}, i.e.,
	\begin{equation}\label{StatisticCSIError}
		\begin{aligned}
			& \Delta \bm{h}_{e,q}\sim \mathcal{CN}(\bm 0, \bm{\Sigma}_{h,q}), ~
			\bm \Sigma_{h,q}\succeq 0, ~\forall q \in \mathcal{Q}, \\
			& \Delta \bm{g}_{e,q}\sim \mathcal{CN}(\bm 0, \bm \Sigma_{g,q}), ~
			\bm \Sigma_{g,q}\succeq 0,  ~\forall q \in \mathcal{Q},
		\end{aligned}
	\end{equation}
	where $  \bm \Sigma_{h,q} \in {{\bf{C}}^{N \times {N}}} $ and $ \bm \Sigma_{g,q} \in {{\bf{C}}^{M \times {M}}} $ are positive semidefinite error covariance matrices for  $\bm{h}_{e,q}^H $ and  ${\bm{g}}_{e,q}^H$, respectively.

	In the following, we  investigate the robust design of transmit beamforming at the PT  and reflection beamforming at the active RIS to minimize  the system power consumption based on the bounded CSI error model and the statistical CSI error model, respectively.

	\section{Robust beamforming design for the bounded CSI error model}
	\label{BoundedCSI}

	In this section, for the bounded CSI error model,  we aim at minimizing the total system power consumption defined in \eqref{TotalPower}  by optimizing the transmit beamforming and  the reflection beamforming  including amplitude reflection coefficients and phase shifts. {The study of this model is conservative, which guarantees the system performance for the worst-case design.}
 The optimization problem for this model is formulated as
%	\begin{equation}\tag{$\textbf{P1}$} 
		{\begin{align}
			\min_{\bm{w},\bm{\Phi}} &~~  \|\bm{w} \|_2^2+M\left( P_{SW}+P_{DC}\right) + ||\bm{\Phi} \bm{G}\bm{w}  ||_2^2 + \sigma_I^2 || \bm{\Phi} ||_F^2 \nonumber \\ 
			\text{s.t.}~~ & \text{C1:~} \gamma_{p,k} \ge \gamma_{s},~ k\in\mathcal{K}, \nonumber \\
			& \text{C2:~} \gamma_{e,q} \le  \gamma_{e}, \nonumber\\
			&{|| \Delta \bm{h}_{e,q}||_2 } \le \xi_{h,q}, ~{|| \Delta \bm{g}_{e,q}||_2 } \le \xi_{g,q},
			~q \in \mathcal{Q}, \nonumber \tag{$\textbf{P1}$}   \\ 
			& \text{C3:~} \gamma_{s}^{(s)}\ge \gamma_{s},\nonumber\\
			& \text{C4:~} \gamma_{s}^{(c)}\ge \gamma_{c},\nonumber\\
			& \text{C5:~}0 \le \beta_m \le \eta, ~0 \le \theta_m < 2 \pi, m \in \mathcal{M},\nonumber
		\end{align}}
%	\end{equation}	
where $\gamma_{s}$ and $\gamma_{c}$ are the decoding SNRs required for the primary and secondary transmissions, respectively,  $\gamma_e$ is the maximum SNR threshold at the Eves. Specifically, the constraint C1 indicates the minimum SNR requirement at the PT for decoding $s(l)$,
 the constraint C2 is used for guaranteeing the secure transmission by controlling the maximum SNR achieved at the Eves under the bounded CSI error model, {the constraint C3 guarantees that the primary signals can be successfully decoded at the SU and then the SIC technique can be implemented correctly \cite{QianqianZhang,Broadcasting},} the constraint C4 shows the minimum SNR requirement at the SU for decoding $c$,  and the constraint C5  indicates the reflection coefficients of the RIS.

 It is obvious that  \textbf{P1}  is a non-convex optimization problem due to the coupling of $ \bm w $ and $ \bm \Phi $ in the  constraints, for which the globally optimal solution is challenging to obtain. We thus propose an AO algorithm consisting of two-layer iterations. Specifically, in the outer iteration, we optimize the transmit beamforming (i.e., $\bm{w}$) and reflection beamforming (i.e., $\bm{\Phi}$) in two sub-problems alternately. While in the inner layer iteration, we propose to use the SROCR technique to accurately find a rank-one solution of the reflection beamforming.
  Before showing the solutions, we first introduce the S-Procedure in Lemma \ref{S-Procedure}, which will be used for the reformulation of the constraint C2 in Section \ref{ActiveBeamBounded} and Section \ref{PassiveBeamBounded}, respectively.

	\begin{lemma}
		\label{S-Procedure}
		(S-Procedure \cite{S-procedure})  Let $\bm F_\tau \in  \mathbb{H}^{n\times n }$, $\bm g_\tau \in  \mathbb{C}^n$, $h_\tau \in  \mathbb{R}$, where $\tau = \{1,2\}$. Suppose there exists a point $\hat{\bm{x}}$ satisfying 
		$$\hat{\bm{x}}^H \bm F_1 \hat{\bm{x}} + 2 \text{Re}(\bm{g}_1^H \hat{\bm{x}} ) +h_1 < 0. $$
Then, the implication 
	\end{lemma}
	\begin{equation*}
		\begin{aligned}
			&~~~~ \bm x^H \bm F_1 \bm x + 2 \text{Re}(\bm {g_1}^H \bm x) +h_1 \le 0 \\& \Rightarrow \bm x^H \bm F_2 \bm x + 2 \text{Re}(\bm {g_2}^H \bm x) +h_2 \le 0 
		\end{aligned}
	\end{equation*}
	 holds if and only if there exists a $ \varpi \ge 0 $ such that
	\begin{equation*}
		\begin{aligned}\label{}
			\varpi &\left[ {\begin{array}{*{20}{c}}
					{{\bm F_1}}& \bm  g_1\\
					\bm {g_1}^H&{ h_1}
			\end{array}} \right] \succeq \left[ {\begin{array}{*{20}{c}}
					{{\bm F_2}}& \bm  g_2\\
					\bm {g_2}^H&{ h_2}		
			\end{array}} \right].\\
		\end{aligned}
	\end{equation*} 
\subsection{Transmit Beamforming Optimization} 
\label{ActiveBeamBounded}
In this subsection, we aim to optimize the transmit beamforming vector $\bm w$ with $\bm{\Phi}$ being fixed for the bounded CSI error model. To make the sub-problem tractable, we first introduce an auxiliary variable for substitution.
We denote $\bm{W} = \bm w \bm w^H$ with Rank($\bm W$)=1 and $\bm W \succeq \bm 0$. 
	% ${\left|  \left( \bm h_{p,k}^H + \bm g_{p,k}^H {\bm \Phi } \bm G \right)\bm w \right|}^2 +{\left|  \left( \bm h_{p,k}^H - \bm g_{p,k}^H {\bm \Phi } \bm G \right)\bm w \right|}^2 $ in C1 can be reduced into $2\left(  \bm h_{p,k}^H \bm W \bm h_{p,k} + \bm g_{p,k}^H {\bm \Phi } \bm G \bm W{ \left( \bm g_{p,k}^H {\bm \Phi } \bm G\right)}^H\right)  $. 
Then, the constraints C1, C3 and C4 can be rewritten as 
\begin{align*}
&\text{C6:} ~\text{Tr}\left( \bm H_{p,k} \bm W\right)+  \text{Tr}  \left( \bm G_{p,k} \bm W\right) \ge \gamma_s  \kappa_{p,k},  ~k\in\mathcal{K}, \\
& \text{C7:}~ \text{Tr}\left( \bm H_{s} \bm W\right)+  \text{Tr}  \left( \bm G_{s} \bm W\right) \ge \gamma_s  \kappa_{s},\\
& \text{C8:} ~ L\text{Tr} \left( \bm G_s \bm W\right) \ge  \gamma_c  \kappa_{s},
\end{align*}
respectively, where 
$\bm H_{p,k}= \bm h_{p,k} \bm {h}_{p,k}^H$, 
$\bm{H}_{s}= \bm{h}_{s} \bm{h}_{s}^H$,
$\bm G_{p,k} ={ \left( \bm g_{p,k}^H {\bm \Phi } \bm G\right)}^H  \left( \bm g_{p,k}^H {\bm \Phi } \bm G \right) $, 
$\kappa_{p,k}= {\sigma_{p,k}^2 + \sigma_I^2|| \bm g_{p,k}^H { \bm \Phi }  ||_2^2 }$,  
$\bm G_{s} ={ \left( \bm g_{s}^H {\bm \Phi } \bm G\right)}^H  \left( \bm g_{s}^H {\bm \Phi } \bm G \right) $, 
 and  $\kappa_{s}= {\sigma_{s}^2 + \sigma_I^2 || \bm g_{s}^H { \bm \Phi }  ||_2^2 }$. Similarly, the constraint C2 can be reformulated as 
%\begin{equation}
		\begin{align}\label{QMI}
			&{\Delta \bm x}_q^H\tilde {\bm W}{\Delta \bm x_q}+{\Delta \bm x}_q^H\tilde {\bm W}\hat{\bm x}_q+\hat{\bm x}_q^H\tilde {\bm W}{\Delta \bm x}_q \\
			&+\hat{\bm x}_q^H\tilde {\bm W}\hat{\bm x}_q \le \gamma_{e}\sigma_{e,q}^2, ~ ||{\Delta {\bm x}_q} ||_2^2\le \xi _{h,q}^2 + \xi _{g,q}^2, ~ q\in\mathcal{Q},\nonumber
		\end{align}
%	\end{equation}
where 
$\tilde {\bm W} = \left({\begin{array}{*{20}{c}}
		{ \bm J_w }& \bm 0\\
		\bm 0&{\bm W}
\end{array}}\right)$,
$ \bm x_q = \left( {\begin{array}{*{20}{c}}
			{{\bm g_{e,q}}}\\
			{{\bm h_{e,q}}}
	\end{array}} \right) $, 
$\hat{\bm x}_q  = \left( {\begin{array}{*{20}{c}}
			{{  \hat{\bm g}_{e,q}}}\\
			{{  \hat{\bm h}_{e,q}}}
	\end{array}} \right) $, 
$\Delta \bm x_q = \left( {\begin{array}{*{20}{c}}
			{{\Delta  \bm g_{e,q}}}\\
			{{\Delta \bm h_{e,q}}}
	\end{array}} \right) $, 
 and $\bm J_w =\bm \Phi \bm G \bm W {\bm G^H  }{\bm \Phi ^H} -\gamma_{e} \sigma^2_I \bm{\Phi}  \bm{\Phi}^H $. In order to tackle the CSI uncertainties in \eqref{QMI}, the S-Procedure described in Lemma \ref{S-Procedure} is used to transform it into an equivalent LMI for $q \in \mathcal{Q}$  as follows:
%	\begin{equation}
		\begin{align}\label{LMIs1}
			\varpi_q &\left[ {\begin{array}{*{20}{c}}
					{{\bm I_{N + M}}}& \bm 0\\
					\bm 0&{ - {{\xi_q}^2}}
			\end{array}} \right] \succeq \left[ {\begin{array}{*{20}{c}}
					\tilde {\bm W}&{\tilde {\bm W} \hat {\bm x}_q}\\
					{ {\hat{\bm x}_q}^H \tilde {\bm W}}&{ {\hat{\bm x}_q}^H \tilde {\bm W}\hat {\bm x}_q-\gamma_{e}\sigma_{e,q}^2 }			
			\end{array}} \right],~q\in\mathcal{Q},
		\end{align}
%	\end{equation}
	{where $ \varpi_q\ge0 $ is a slack variable for the $q$-th Eve and needs to be optimized in the following,}
 $\xi_q^2 =\xi _{h,q}^2 + \xi _{g,q}^2$. It is noted that \eqref{LMIs1} is derived from \eqref{QMI} by setting the parameters in Lemma \ref{S-Procedure} as follows: $\bm x= \Delta \bm x_q$, $\bm F_1=\bm I_{N + M} $, $\bm g_1= \bm 0$, $h_1=-\xi_q^2 $,  $\bm F_2=\tilde {\bm W} $, $\bm g_2= \tilde {\bm W} \hat {\bm x}_q $, and $h_2={\hat{\bm x}_q}^H \tilde {\bm W}\hat {\bm x}_q-\gamma_{e}\sigma_{e,q}^2 $. {The detailed derivation of \eqref{LMIs1} can be found in \cite{S-procedure}.}
Then, we derive   the equivalent form of the constraint C2, which is given by
$$\text{C9:~} \hat{\bm{W}_q}\succeq  \bm 0, ~q \in \mathcal{Q},$$
$$\hat{\bm{W}_q} = \left[ {\begin{array}{*{20}{c}}
		{{ \varpi_q \bm I_{N + M}-\tilde {\bm W}}}&  -\tilde {\bm W} \hat {\bm x}_q \\
		-{\hat{\bm x}_q}^H \tilde {\bm W}&{ - \varpi_q{{\xi}_q^2} -{\hat{\bm x}_q}^H \tilde {\bm W}\hat {\bm x}_q+\gamma_{e}\sigma_{e,q}^2 }
\end{array}} \right],$$
%\begin{align}
%	 \text{C9:~} \hat{\bm{W}} = \left[ {\begin{array}{*{20}{c}}
%			{{ \varpi_q \bm I_{N + M}-\tilde {\bm W}}}&  -\tilde {\bm W} \hat {\bm x}_q \\
%			-{\hat{\bm x}_q}^H \tilde {\bm W}&{ - \varpi_q{{\xi}_q^2} -{\hat{\bm x}_q}^H \tilde {\bm W}\hat {\bm x}_q+\gamma_{e}\sigma_{e,q}^2 }
%	\end{array}} \right] \succeq  \bm 0, ~q \in \mathcal{Q},
%\end{align}	
where $\hat{\bm{W}_q} \in \mathcal{C}^{ (M+N+1)\times (M+N+1) }$. {Hence, the sub-problem for optimizing $\bm W$ can be given by
	\begin{equation}\tag{$\textbf{P2}$}\label{P2} 
		\begin{aligned}
			\min_{\bm{W}, \bm{\varpi}} ~~& \text{Tr}\left( \bm W \right) + \text{Tr} (\bm{\Upsilon} \bm{W})\\ 
			\text{s.t.}~~ & \text{C6}, \text{C7}, \text{C8}, \text{C9}, \\
			& \text{C10:~} \text{Rank}(\bm W)=1,\\
			& \text{C11:~} \bm W \succeq \bm 0,
		\end{aligned}
	\end{equation}	
	where $\bm{\varpi} = [\varpi_1,\ldots,\varpi_Q]$, and $\bm{\Upsilon}  = \bm{G}^H \bm{\Phi}^H \bm{\Phi} \bm{G}$.}
However, \textbf{P2} is still a non-convex optimization problem due to the rank-one constraint shown in C10. To address this issue,  the SDR technique \cite{Luo} is adopted to relax the constraint C10, and \textbf{P2} is then reduced to \textbf{P3}, which is given by{ 
	\begin{equation}\tag{$\textbf{P3}$}\label{P3} 
		\begin{aligned}
			\min_{\bm{W},\bm{\varpi}} ~~& \text{Tr}\left( \bm W \right) + \text{Tr} (\bm{\Upsilon} \bm{W}) \\ 
			\text{s.t.}~~ & \text{C6}, \text{C7}, \text{C8}, \text{C9}, \text{C11}.
		\end{aligned}
	\end{equation}}
	\textbf{P3} now becomes a semidefinite program (SDP), which can be solved by the interior-point method \cite{S-procedure}. However, the optimal solution to \textbf{P3}, i.e., $ \bm W^* $, achieved by using the SDR may not satisfy the rank-one constraint in C10. In the following theorem, we prove that $ \bm{W}^* $ is always a rank-one matrix.
	\begin{theorem}
		\label{TheoremRankOne}
		We can always find a rank-one transmit beamforming matrix as the optimal solution to \textbf{P3}.
	\end{theorem}
	
	\begin{proof}
		Please refer to Appendix.
	\end{proof}
Theorem \ref{TheoremRankOne} guarantees the tightness of applying the SDR, i.e., the optimal solution to \textbf{P3} is also the optimal solution to \textbf{P2}. Then, we can obtain the optimal transmit beamforming vector $\bm{w}^{(r)}$ from $\bm{W}^*$ by applying the Cholesky decomposition, where the superscript $r$ denotes the $r$-th outer iteration for the AO algorithm.

\subsection{Reflection Beamforming Optimization }
	\label{PassiveBeamBounded}
	 Given $ \bm{W}^* $ (i.e., $\bm{w}^{(r)}$) in the $r$-th outer iteration, we proceed to optimize the reflection beamforming at the RIS for the bounded CSI error model. Similarly, we let $ \bm V= \bm v \bm v^H $, where $\bm{v} = [v_1, \ldots, v_M]^T$, $ \bm V \succeq \bm  0 $, $\text{Rank}(\bm V)=1$, and $\bm V(m,m) \le \eta^2, m \in \mathcal{M}$. Then, C1, C3, and C4 can be respectively recast as
\begin{align*}
& \text{C12:}~ \text{Tr}\left( \bm H_{p,k} \bm W\right)+  \text{Tr}  \left(  {\tilde{\bm G}_{p,k}} \bm V\right) - \tilde {\kappa}_{p,k} \gamma_s \ge   \sigma_{p,k}^2 \gamma_{s}, k\in\mathcal{K},\\
& \text{C13:}~  \text{Tr}\left( \bm H_{s} \bm W\right)+  \text{Tr}  \left(  {\tilde{\bm G}_{s}} \bm V\right) - \tilde {\kappa}_{s} \gamma_s  \ge \sigma_s^2 \gamma_{s}, \\
& \text{C14:}~L \cdot \text{Tr}  \left(  {\tilde{\bm G}_{s}} \bm V\right) - \tilde {\kappa}_{s} \gamma_c \ge  \sigma_s^2 \gamma_{c} ,
\end{align*}
	where ${\tilde{\bm G}_{p,k}} = {\text{diag} \left( \bm G \bm w \right)}^H \bm g_{p,k} \bm g_{p,k}^H \text{diag} \left( \bm G \bm w \right) $,  $  \tilde {\kappa}_{p,k}=  \sigma_{I}^2  \bm g_{p,k}^H \bm V  \bm g_{p,k}  $, ${\tilde{\bm G}_{s}}={\text{diag} \left( \bm G \bm w \right)}^H \bm g_{s} \bm g_{s}^H \text{diag} \left( \bm G \bm w \right)$, and $\tilde {\kappa}_{s}= \sigma_{I}^2  \bm g_{s}^H \bm V  \bm g_{s} $.
	
 Herein, the S-Procedure is also used to transform the constraint C2 into equivalent LMIs. By letting $ \tilde  {\bm V}= \left({\begin{array}{*{20}{c}}
			{ \bm J_v }& \bm 0\\
			\bm 0&{\bm W}
	\end{array}}\right) $ with $ \bm J_v={\text{diag} \left( \bm G \bm w \right)} \bm V {\text{diag} \left( \bm G \bm w \right)}^H-\sigma_I^2 \gamma_e   \bm V $,  the constraint C2 can be  recast as follows:
\begin{align}
	\text{C15:~} \left[ {\begin{array}{*{20}{c}}
			{{ \omega_q \bm I_{N + M}-\tilde {\bm V}}}&  -\tilde {\bm V} \hat {\bm x}_q \\
			-{\hat{\bm x}_q}^H \tilde {\bm V}&{ - \omega_q {{\xi}_q^2} -{\hat{\bm x}_q}^H \tilde {\bm V}\hat {\bm x}_q+\gamma_{e}\sigma_{e,q}^2 }
	\end{array}} \right] \succeq  \bm 0,~q \in \mathcal{Q},\nonumber
\end{align}
where $ \omega_q $ denotes a slack variable.
	Thus, the sub-problem for optimizing $   {\bm V}$ is written as 
	\begin{equation}\tag{$\textbf{P4}$} 
		{ \begin{aligned}
			  \min_{\bm{V}, \bm{\omega}}~~& \text{Tr} ( \bar{\bm{\Upsilon}}  \bm{V} )  + \sigma_I^2 \text{Tr}(\bm{V})  \\ 
			\text{s.t.}~~ & \text{C12}, \text{C13}, \text{C14}, \text{C15},   \\
			& \text{C16:~} \text{Rank}(\bm V)=1,\\
			&\text{C17:~} \bm V \succeq \bm 0, \\ 
			& \text{C18:~} \bm V(m,m) \le \eta^2, m \in \mathcal{M},
		\end{aligned}}
		\end{equation}	
where $\bm{\omega} = [\omega_1, \ldots, \omega_Q]$, and $\bar{\bm{\Upsilon}}  = \text{diag} (\bm{G}\bm{w} ) \text{diag} (\bm{G}\bm{w} )^H$. It is noted that the objective value of \textbf{P4} may be quite small and comparable with the threshold predefined for the proposed AO algorithm.
To achieve a better convergence performance, we introduce the SNR residuals for the  constraints C12-C14, which are denoted by   $\alpha_{k}, k=1,2,...,K+2$.  Define $\bm{\alpha}=[\alpha_1,\ldots,\alpha_{K+2}]$.
Specifically, \textbf{P4} can be transformed into \textbf{P5}, which is given by
 	\begin{equation}\tag{$\textbf{P5}$} 
		{\begin{aligned}
			&  \max_{\bm{V},\bm{\omega},\bm{\alpha}}~ -\text{Tr} ( \bar{\bm{\Upsilon}}  \bm{V} )  - \sigma_I^2 \text{Tr}(\bm{V}) + \sum_{k=1}^{K+2} \alpha_k \\ 
			\text{s.t.}~~ & \bar{\text{C12}}:~ \text{Tr}\left( \bm H_{p,k} \bm W\right)+  \text{Tr}  \left(  {\tilde{\bm G}_{p,k}} \bm V\right) \\&~~~~~~- \tilde {\kappa}_{p,k} \gamma_s  \ge \left( \gamma_{s}+\alpha_k\right) \sigma_{p,k}^2, ~ k\in\mathcal{K},  \\ 
			&\bar{\text{C13}}:~\text{Tr}\left( \bm H_{s} \bm W\right)+  \text{Tr}  \left(  {\tilde{\bm G}_{s}} \bm V\right) 
			\\&~~~~~~- \tilde{\kappa}_s \gamma_s \ge \left( \gamma_{s} +\alpha_{K+1}\right)  \sigma^2_{s}, \\ 
			& \bar{\text{C14}}:~ L  \text{Tr}  \left(  {\tilde{\bm G}_{s}} \bm V\right) - \tilde{\kappa}_s \gamma_c \ge \left( \gamma_{c}+\alpha_{K+2} \right) \sigma^2_{s}, \\
			 & \text{C15},~ \text{C16},~ \text{C17},~ \text{C18}.
		\end{aligned}}
	\end{equation}

After using the SDR technique to relax the constraint C17, we can also use the interior-point method to solve \textbf{P5}.
 However, due to the existence of the constraint {C16}, the obtained numerical solution from \textbf{P5} by applying the SDR is generally not rank-one \cite{Derrick}. Thus,  we propose to use the SROCR  technique \cite{Poor} to deal with the rank-one constraint C16. 
 It is worth noting that the  constraint  C16 is equivalent to the following equality, i.e.,
$\rho_\text{max}(\bm{V}) = \text{Tr}(\bm{V}),$
 where $\rho_\text{max}(\bm{V})$ represents the largest eigenvalue of $\bm{V}$. In addition, $\rho_\text{max}(\bm{V})$ is equal to the optimal result of the following problem
 \begin{align*}
\rho_\text{max}(\bm{V}) = \max_{\bm{\varphi} \in \mathcal{C}^{M \times 1}, ||\bm{\varphi}||_2\le 1} \bm{\varphi}^H \bm{V} \bm{\varphi},
 \end{align*}
where $\bm{\varphi}$ is a slack vector. Thus, we can have 
\begin{align}
\label{RankOneNew}
\text{Tr}(\bm{V}) = \max_{\bm{\varphi} \in \mathcal{C}^{M \times 1}, ||\bm{\varphi}||_2 \le 1} \bm{\varphi}^H \bm{V} \bm{\varphi},
\end{align} 
and the constraint C16 can be replaced with \eqref{RankOneNew}.  It is obvious that the optimal solution to \eqref{RankOneNew}  is the largest eigenvector of $\bm{V}$. However, \eqref{RankOneNew} is still a non-convex constraint.
To transform \eqref{RankOneNew} into a convex one,  for any feasible point $(\bm{V}^{(r,t)}, \omega^{(r,t)})$ in the $r$-th outer iteration  and the $t$-th inner iteration for the SROCR technique, we can partially relax \eqref{RankOneNew} (i.e., C16) to 
\begin{align}
	\label{RankOneRelaxation}
	\text{C19:~} \bm{\vartheta} (\bm{V}^{(r,t)})^H \bm{V} \bm{\vartheta} (\bm{V}^{(r,t)}) \ge \omega^{(r,t)} \text{Tr} (\bm{V}),
\end{align}
where $\bm{\vartheta}$ denotes the largest eigenvector of $\bm{V}^{(r,t)}$, 
$\omega^{(r,t)}  \le 1$ is a slack parameter controlling the ratio of the largest eigenvalue to the trace for $\bm{V}$, and the equality holds if $\omega^{(r,t)}=1$, under which the rank-one constraint is satisfied. 
With \eqref{RankOneRelaxation}, \textbf{P5} can be reformulated as 
 	\begin{equation}\tag{$\textbf{P6}$} 
		\begin{aligned}
			&  \max_{\bm{V}, \bm{\omega}, \bm{\alpha}}~ -\text{Tr} ( \bar{\bm{\Upsilon}}  \bm{V} )  - \sigma_I^2 \text{Tr}(\bm{V}) + \sum_{k=1}^{K+2} \alpha_k \\ 
			\text{s.t.}~~ & \bar{\text{C12}}, ~\bar{\text{C13}}, \bar{\text{C14}}, \text{C15}, \text{C17}, \text{C18}, \text{C19}.
		\end{aligned}
	\end{equation}
	It is readily observed that \textbf{P6} is convex and can be solved by using the interior-point  method. We then update $\omega^{(r,t)}$ sequentially from 0 to 1, the process of which is given by \cite{Poor}
\begin{align}
\label{StepSize}
	\omega^{(r,t+1)} = \min\left(\frac{\varphi(\bm{V}^{(r, t)})}{\text{Tr}(\bm{V}^{(r,t)})}+\delta^{(r,t)},1 \right),
\end{align} 
where  $ \varphi({\bm{V}^{(r,t)}})$ denotes the largest eigenvalue of $\bm{V}^{(r,t)}$, and $\delta^{(r,t)}$ is the step size.  

By iteratively optimizing $\textbf{P6}$, we can finally find the sub-optimal solution, denoted by  $\bm{V}^*$, to \textbf{P5}. Then,   the sub-optimal reflection beamforming vector in the $r$-th outer iteration, denoted by ${\bm{v}}^{(r)}$, can be obtained by implementing the Cholesky decomposition of $\bm{V}^*$.
\subsection{Algorithm Summary and Analysis}
% The total algorithm to solve \textbf{P1} is summarized in Algorithm \ref{AlgorithmA}. Since the objective function of \textbf{P1} is a non-decreasing function after each iteration, the convergence of Algorithm \ref{AlgorithmA} can be guaranteed, which will be  confirmed by numerical simulations in Section \ref{NumericalResults}. We then analyze the computational complexity of Algorithm \ref{AlgorithmA}. It is obvious that its computational complexity is mainly caused by solving \textbf{P3} and \textbf{P6}, the complexities of which are $\mathcal{O}(\sqrt{N} \log(1/\rho_1) ( (Q+K+2)N^3 +(Q+K+2)^2N^2 + (Q+K+2)^3  ) )$ and $\mathcal{O}(\sqrt{M} \log(1/\rho_2) ( (Q+K+3)M^3 +(Q+K+3)^2 M^2 + (Q+K+3)^3  ) )$ \cite{Complexity}, respectively, where $\rho_1$ and $\rho_2$ are the convergence tolerance thresholds of \textbf{P3} and \textbf{P6}, respectively, and $\mathcal{Q}$ denotes the big-O notation. Thus, the total complexity of Algorithm \ref{AlgorithmA} is $ \mathcal{O}(I_1 \sqrt{N} \log(1/\rho_1) ( (Q+K+2)N^3 +(Q+K+2)^2N^2 + (Q+K+2)^2N  ) ) +  \mathcal{O}(I_1 I_2\sqrt{M} \log(1/\rho_2) ( (Q+K+3)M^3 +(Q+K+3)^2 M^2 + (Q+K+2)^2 M  ) )$, where $I_1$ and $I_2$ denote the maximum times of the AO algorithm and the SROCR technique, respectively. 

The  algorithm to solve \textbf{P1} is summarized in Algorithm \ref{AlgorithmA}. Since the objective function of \textbf{P1} is a non-increasing function after each iteration, the convergence of Algorithm \ref{AlgorithmA} can be guaranteed. {The convergence of Algorithm \ref{AlgorithmA} can be proved in a similar way as that in  \cite{Convergence}. Moreover, we also  provide numerical simulations in Section \ref{NumericalResults} to confirm the convergence performance.} We then analyze the computational complexity of Algorithm \ref{AlgorithmA}. It is obvious that its computational complexity is mainly caused by solving \textbf{P3} and \textbf{P6}, the complexities of which are $\mathcal{O}( (Q(M+N+1)+N)^{1/2} N^2(N^4 + N^2 Q (M+N+1)^2 + N^4 +Q(M+N+1)^3+ N^3 ) )$ and $\mathcal{O}( (Q(N+M+1)+M)^{1/2} M^2(M^4 + M^2 Q (N+M+1)^2 + M^4 +Q(N+M+1)^3+ M^3 ) )$ \cite{Complexity}, respectively, and $\mathcal{O}$ denotes the big-O notation. Thus, the total complexity of Algorithm \ref{AlgorithmA} is $ \mathcal{O}(  I_1( (Q(M+N+1)+N)^{1/2} N^2(N^4 + N^2 Q (M+N+1)^2 + N^4 +Q(M+N+1)^3+ N^3 ) ) + I_1 I_2  ( (Q(N+M+1)+M)^{1/2} M^2(M^4 + M^2 Q (N+M+1)^2 + M^3 +Q(N+M+1)^3+ M^3 ) ))$, where $I_1$ and $I_2$ denote the maximum iteration times of the AO algorithm and the SROCR technique used in Algorithm \ref{AlgorithmA},  respectively.

	\begin{algorithm}
	\caption{ The proposed AO Algorithm for Problem \textbf{P1}}
	\label{AlgorithmA}
	\begin{algorithmic}[1]	
			\STATE {Initialization: Convergence tolerance $ \varepsilon$, outer iteration index $r=1$, and reflection beamforming vector $\bm{v}^{(r)}$.} 
			\REPEAT
			
			\STATE{Update $\bm{W}^{*}$ by solving \textbf{P3} and obtain $\bm{w}^{(r)}$ by applying the Cholesky decomposition of $\bm{W}^{*}$.} 
			\STATE{Initialization: Inner iteration index $t=0$.}
				\REPEAT 	
				\STATE{Update $\bm{V}^{(r,t+1)}$ by solving \textbf{P6}. }
				\STATE{Update $\omega^{(r,t+1)}$ based on \eqref{StepSize}.}
			
				\UNTIL{ The convergence is achieved.}   
				\STATE{Set $\bm{V}^* =\bm{V}^{(r,t+1)}$ and compute $\bm{v}^{(r)}$ accordingly.}
				\STATE{$r=r+1$.}
			\UNTIL{ the decrease of the objective function  value in \textbf{P1} is below $ \varepsilon$. }
				\STATE{Set $\bm{w}^* = \bm{w}^{(r)}$ and $\bm{v}^* = \bm{v}^{(r)}$ for $m\in \mathcal{M}$.}
		\end{algorithmic}
	\end{algorithm}
	
	\section{Robust beamforming design for The Statistical CSI Error Model}
	\label{StatisticalCSI}
It is known that the bounded CSI error model is typically used for modeling the quantization error. However, the channel estimation error tends to be  Gaussian distribution, which is unbounded \cite{Pan}. Hence, {in this section, we consider the  statistical CSI error model for investigating the impact of estimation errors on system performance.} Different from the worst-case SNR constraint, we consider the  SNR outage constraint at the Eves for the statistical CSI error model, which are defined as follows
	\begin{align}
&\text{C20:~} \text{Pr} \left\{ {\gamma_{e,q} \le  \gamma_{e} } \right\} \ge 1-\rho_q, \\&~\Delta \bm{h}_{e,q}\sim \mathcal{CN}(\bm 0, \bm \Sigma_{h,q}),~ \Delta \bm{g}_{e,q}\sim \mathcal{CN}(\bm 0, \bm \Sigma_{g,q}),~q\in\mathcal{Q},\nonumber
	\end{align}
	where $\rho_q$ is a constant representing the outage probability at the $q$-th Eve.
Under this constraint, the system power consumption minimization problem for the statistical CSI error model can be formulated as  follows
%	\begin{equation}\tag{$\textbf{P7}$} 
		\begin{align}
			\min_{\bm{w},\bm{\Phi}} ~~&  \|\bm{w} \|_2^2+M\left( P_{SW}+P_{DC}\right) + ||\bm{\Phi} \bm{G}\bm{w}  ||_2^2 + \sigma_I^2 || \bm{\Phi} ||_F^2 \nonumber \\ 
			\text{s.t.}~~ & \text{C1}, ~\text{C3},~\text{C4},~ \text{C5},~\text{C20}.\nonumber \tag{$\textbf{P7}$} 
		\end{align}
%		{\color{blue}\begin{aligned}			
%		\end{aligned}}
%	\end{equation}	
	Similar to \textbf{P1}, \textbf{P7} is  a non-convex optimization problem. To address the non-convexity, an AO algorithm with the Bernstein-Type Inequality, SROCR and SDR techniques is proposed to optimize $\bm{w}$ and $\bm{\Phi}$  in an alternating manner. In the following lemma, we first introduce the Bernstein-Type Inequality \cite{Bernstein-Type}, which will be used to reformulate the constraint C20.

	\begin{lemma}
	\label{LemmaBTI}
		(Bernstein-Type Inequality \cite{Bernstein-Type}) Considering $f(\bm x)=\bm x^H \bm{B} \bm x + 2\text{Re}\left\lbrace {\bm{b}^H \bm x}\right\rbrace +b$, where $ \bm{B} \in \mathbb{H}^{r\times r} $, $\bm{b} \in \mathbb{C}^{r \times 1}$, $b \in \mathbb{R}$,  and $ \bm{x} \in \mathbb{C}^{r \times 1}  \sim \mathcal{CN}(\bm 0, \bm{I}_r)$. If there exists a variable satisfying $\rho \in [0,1]$, the following approximation can be obtained
		\begin{equation}
			\begin{aligned}\label{BTI1}
				&\text{Pr} \left\{\bm x^H \bm B \bm x +2 \text{Re} \left\{\bm{b}^H \bm{x} \right\} + b \ge 0 \right\} \ge 1-\rho
			\end{aligned}
		\end{equation}
		\begin{equation}
			\begin{aligned}\label{BTI2}
				&\Rightarrow \text{Tr}\left\{ \bm B  \right\}-\sqrt{2\ln\left( 1/{\rho }\right)   } \sqrt{{ { \left| \left|  \bm B   \right| \right|}_F^2  }+ 2 ||  \bm b   ||_2^2            } \\
				&~~~~~~~~~~+\ln\left( {\rho }\right)\lambda_{max}^+(-\bm B) + {b}\ge0    
			\end{aligned}
		\end{equation}
		\begin{equation}
			\begin{aligned}\label{BTI}
				\Rightarrow \left\{ {\begin{array}{*{20}{c}}
						& \text{Tr}\left\{ \bm{B}  \right\}-\sqrt{2 \ln\left( 1/{\rho}\right)   }{e}+\ln\left( {\rho}\right){f}+ {b} \ge  0\\
						&\sqrt{{ { \left| \left|  \bm B   \right| \right|}_F^2  }+ 2 ||  \bm b   ||_2^2            }\le e\\
						&f\bm I_r + \bm B \succeq \bm 0, f \ge 0, 
				\end{array}} \right.
			\end{aligned}
		\end{equation}
		where $e$ and $f$  are	slack variables, 
		$\lambda_{\max}^+(- \bm B)= \max(\lambda_{\max}(- \bm  B),0)$. 
	\end{lemma}

\subsection{Transmit  Beamforming Optimization}
\label{ActiveBeamStat}
Given $\bm{\Phi}$, we optimize the transmit beamforming for the statistical CSI error model.
Similar to Section \ref{ActiveBeamBounded}, by introducing $\bm{W}$ defined in Section \ref{ActiveBeamBounded} for substitution, the constraint C20 can be equivalently transformed as 
	\begin{equation}
		\begin{aligned}\label{inequaw}
		&\text{Pr} \left\{ {\sigma_{e,q}^2\gamma_{e}-  \left( \bm h_{e,q}^H \bm W \bm h_{e,q} + \bm g_{e,q}^H \bm J_w \bm g_{e,q}\right)\ge 0 } \right\} \ge 1-\rho_q, \\ 
		&\Delta \bm{h}_{e,q}\sim \mathcal{CN}(\bm 0, \bm \Sigma_{h,q}),~\Delta \bm{g}_{e,q}\sim \mathcal{CN}(\bm 0, \bm \Sigma_{g,q}),~q\in\mathcal{Q}.
		\end{aligned}
	\end{equation}
However, since there exist no closed-form expressions for the constraints in \eqref{inequaw}, it is computationally intractable to solve the transmit beamforming optimization problem under \eqref{inequaw}. To deal with this challenge, the Bernstein-Type Inequality \cite{Bernstein-Type} in Lemma \ref{LemmaBTI} can be utilized to approximate \eqref{inequaw} in a safe manner.  {According to \cite{Bernstein-Type},  to simplify the derivation and find a convex and computationally tractable approximation on \eqref{inequaw},   we assume that $\bm \Sigma_{h,q}=\epsilon_{h,q}^2 \bm I_{N}$ and  $\bm \Sigma_{g,q}=\epsilon_{g,q}^2 \bm I_M $, which has been widely used in the literature \cite{PanTwo,Pan}.} Then, we have $  \Delta h_{e,q}=\epsilon_{h,q} \bm i_{h,q} $,  $  \Delta g_{e,q}=\epsilon_{g,q} \bm i_{g,q} $, where $ \bm i_{h,q} \sim \mathcal{CN}(\bm 0, \bm I_{N}) $ and $ \bm i_{g,q} \sim \mathcal{CN}(\bm 0, \bm I_{M}) $.
 \eqref{inequaw} can be thus constructed as the form of \eqref{BTI1}, which is given by 
 	\begin{equation}
		\begin{aligned}\label{BTIw}
			\text{Pr} \left\{\bm i_q^H \bm B_q \bm i_q +2 \text{Re} \left\{\bm b_q^H\bm i_q \right\} +b_q\ge 0 \right\} \ge 1-\rho_q,
		\end{aligned}
	\end{equation}
	where 
	$\bm{b}_q^H = \left[ {\begin{array}{*{20}{c}}
			{ -\epsilon_{g,q}\hat {\bm{g}}_{e,q} \bm J_w  }&{ -\epsilon_{h,q} \hat {\bm{h}}_{e,q} \bm W }
	\end{array}} \right]$,
 	$b_q=\sigma_{e,q}^2 \gamma_e- \hat {\bm{h}}_{e,q}^H \bm W  \hat {\bm{h}}_{e,q} - \hat {\bm{g}}_{e,q}^H \bm J_w  \hat {\bm{g}}_{e,q} $, 
	${\bm B_q} = \left[ {\begin{array}{*{20}{c}}
			{ - \epsilon_{g,q}^2 \bm J_w}& \bm 0\\
			\bm 0&{ - \epsilon_{h,q}^2 \bm W}
	\end{array}} \right]$,
	  and $ \bm i_q = \left( {\begin{array}{*{20}{c}}
			{{\bm i_{g,q}}}\\
			{{\bm i_{h,q}}}
	\end{array}} \right) $.
By applying the Bernstein-Type Inequality and introducing slack variables $\bm{e} = [e_1,\ldots,e_Q]$ and $\bm{f}=[f_1,\ldots,f_Q]$, the following approximation of \eqref{BTIw} in the deterministic form can be derived
\begin{align}
\label{C2Bw1}
	&\text{Tr}\left\{ \bm B_q  \right\}-\sqrt{2 \ln\left( 1/{\rho _q}\right)   }{e_q}+\ln\left( {\rho _q}\right){f_q}+ {b_q}\ge 0, ~q \in \mathcal{Q},  \\
	\label{C2Bw2}
	&\sqrt{{ { \left| \left|  \bm B_q   \right| \right|}_F^2  }+ 2 ||  \bm b_q   ||_2^2            }\le e_q, ~q \in \mathcal{Q},    \\
	\label{C2Bw3}
	&f_q\bm I_{N+M}+ \bm B_q  \succeq \bm 0, ~f_q \ge 0, ~q \in \mathcal{Q}. 
\end{align}

% {\color{red}The derivation process from (28) to (30) can be added.}

Then, by implementing  mathematical transformations, \eqref{C2Bw2} can be further constructed in a simplified form as 
\begin{align}  
\label{C2Bw2New}  
			\left\| {\begin{array}{*{20}{c}}
					\epsilon_{g,q}^2 \text{vec}\left(\bm J_w \right) \\
					\epsilon_{h,q}^2 \text{vec}\left(\bm W \right)\\
					\sqrt{2}\epsilon_{g,q}  \bm J_w \hat {\bm{g}}_{e,q}\\
					\sqrt{2}\epsilon_{h,q}  \bm W \hat {\bm{h}}_{e,q}
			\end{array}} \right\|_2 \le {e_q}, ~q \in \mathcal{Q}.
\end{align}
Similar to \textbf{P2}, the sub-problem for optimizing the transmit beamforming under the statistical CSI error model can be formulated as 
	\begin{equation}\tag{$\textbf{P8}$}\label{P7} 
		{\begin{aligned}
			\min_{\bm{W},\bm{e},\bm{f}} ~~& \text{Tr}\left( \bm W \right) + \text{Tr} (\bm{\Upsilon} \bm{W})\\
			\text{s.t.}~~ & \text{C6}, \text{C7}, \text{C8}, \text{C10}, \text{C11},\\
			& \eqref{C2Bw1},~ \eqref{C2Bw3}, ~\eqref{C2Bw2New}.
		\end{aligned} }
	\end{equation}
By adopting the SDR technique to relax the constraint C10 in \textbf{P8}, we can use the interior-point method to solve it. Recall that the tightness of applying the SDR is guaranteed in Section \ref{ActiveBeamBounded}, we can also prove that the obtained solution $\bm{W}^*$ to the relaxed version of \textbf{P8} is rank-one. Then, we obtain the optimal transmit beamforming vector $\bm{w}^{(r)}$  in the $r$-th outer iteration for the statistical CSI error model.

\subsection{Reflection Beamforming Optimization }  
Herein, with the obtained $\bm{W}^*$ (i.e., $\bm{w}^{(r)}$) for the $r$-th outer iteration in Section \ref{ActiveBeamStat}, we focus on the optimization of reflection beamforming  for the statistical CSI error model. 
By introducing $\bm{V}$ defined in Section \ref{PassiveBeamBounded}, the constraint C20 can be rewritten as 
\begin{equation}
	\begin{aligned}\label{inequav}
		&\text{Pr} \left\{ {\sigma_{e,q}^2\gamma_{e}-  \left( \bm h_{e,q}^H \bm W \bm h_{e,q} + \bm g_{e,q}^H \bm J_v \bm g_{e,q}\right)\ge 0 } \right\} \ge 1-\rho_q,\\
		&\Delta \bm{h}_{e,q}\sim \mathcal{CN}(\bm 0, \bm \Sigma_{h,q}),~ \Delta \bm{g}_{e,q}\sim \mathcal{CN}(\bm 0, \bm \Sigma_{g,q}),~q\in\mathcal{Q}.
	\end{aligned}
\end{equation}
Similarly, by applying the Bernstein-Type Inequality, \eqref{inequav} can be approximated as follows
\begin{align}
\label{C21First}
			& \text{Tr}\left\{ \tilde {\bm B}_q  \right\}-\sqrt{2 \ln\left( 1/{\rho _q}\right)   }{\tilde e_q}+\ln\left( {\rho _q}\right){\tilde f_q}+ \tilde{b}_q \ge 0,~q \in \mathcal{Q},        \\
			\label{C21Second}
			&~~~~~\left\| {\begin{array}{*{20}{c}}
					\epsilon_{g,q}^2 \text{vec}\left(\bm J_v \right) \\
					\epsilon_{h,q}^2 \text{vec}\left(\bm W \right)\\
					\sqrt{2}\epsilon_{g,q}  \bm J_v \hat {\bm{g}}_{e,q}\\
					\sqrt{2}\epsilon_{h,q}  \bm W \hat {\bm{h}}_{e,q}
			\end{array}} \right\|_2 \le \tilde{e}_q, ~q \in \mathcal{Q},\\
			\label{C21Third}
			&\tilde f_q \bm I_{N+M}+ \tilde {\bm B}_q  \succeq \bm 0, \tilde f_q\ge 0, ~q \in \mathcal{Q},
	\end{align}
	where $ \tilde {\bm B}_q = \left[ {\begin{array}{*{20}{c}}
			{ - \epsilon_{g,q}^2 \bm J_v}& \bm 0\\
			\bm 0&{ - \epsilon_{h,q}^2 \bm W}
	\end{array}} \right]$, $\tilde{\bm e} = [\tilde{e}_1,\ldots,\tilde{e}_Q]$, $ \tilde b_q=\sigma_{e,q}^2 \gamma_e - \hat {\bm{h}}_{e,q}^H \bm W  \hat {\bm{h}}_{e,q} - \hat {\bm{g}}_{e,q}^H \bm J_v  \hat {\bm{g}}_{e,q} $,  and $\tilde{\bm f} = [\tilde{f}_1,\ldots,\tilde{f}_Q]$ are slack variables.
	 Then, the sub-problem for optimizing the reflection beamforming can be recast as
	\begin{equation}\tag{$\textbf{P9}$} 
		{\begin{aligned}
			&  \min_{\bm{V}, \tilde{\bm{e}}, \tilde{\bm{f}} }  \text{Tr} ( \bar{\bm{\Upsilon}}  \bm{V} )  + \sigma_I^2 \text{Tr}(\bm{V}) \\ 
			\text{s.t.}~~ & \text{C12}, \text{C13}, \text{C14}, \text{C16}, \text{C17}, \text{C18}, \eqref{C21First}, \eqref{C21Second}, \eqref{C21Third}.
		\end{aligned}}
	\end{equation}
 To speed up  the proposed AO algorithm's convergence,  we also introduce $\alpha_{k}$ to represent the SNR residuals in the constraints C12, C13 and C14, where $k=1,\ldots,K+2$. In addition, we  introduce the slack variable $\varsigma_q$ in \eqref{C21First}, which is updated by 
\begin{align}
\label{C21Updated}
\text{Tr}\left\{ \tilde {\bm B}_q  \right\}-\sqrt{2 \ln\left( 1/{\rho _q}\right)   }{\tilde e_q}+\ln\left( {\rho _q}\right){\tilde f_q}+ \tilde{b}_q \ge  \varsigma_q, ~q \in \mathcal{Q}.  
\end{align}
 Moreover, the SROCR method is applied to deal with the convexity of the constraint C16, the relaxation   of which is given in the constraint C19. To this end, \textbf{P9} can be reformulated as follows
%\begin{equation}\tag{$\textbf{P10}$} 
%		{\color{blue}\begin{aligned}
	\begin{align}	
			&\max_{\bm{V}, \tilde{\bm{e}}, \tilde{\bm{f}}, \bm{ \alpha}, \bm{\varsigma}} ~~ -\text{Tr} ( \bar{\bm{\Upsilon}}  \bm{V} )  - \sigma_I^2 \text{Tr}(\bm{V}) +  \sum\limits_{k=1}^{K+2} \alpha_k + \sum_{q=1}^Q \varsigma_q  \nonumber\\
			&~~~~\text{s.t.}~
			\bar{\text{C12}}, \bar{\text{C13}}, \bar{\text{C14}},  \text{C17}, \text{C18}, \text{C19}, \eqref{C21Second}, \eqref{C21Third}, \eqref{C21Updated},\nonumber \tag{$\textbf{P10}$} 
	\end{align}	
%		\end{aligned}}
%	\end{equation}
	where $\bm{\varsigma}=[\varsigma_1,\ldots,\varsigma_Q]$.
Again, the interior-point method is applied  to solve \textbf{P10}, and the sub-problem solution to \textbf{P9} can be obtained by iteratively optimizing \textbf{P10} until the convergence is achieved.

The process of solving \textbf{P7} is similar to Algorithm \ref{AlgorithmA} and  neglected here for simplicity. 
{Similarly, the convergence of solving \textbf{P7} can also be guaranteed \cite{Convergence},} and
the complexity  is $\mathcal{O} ( I_3 (N+  (N+M)Q+2Q  )^{1/2} N^2 (N^4 + N^2Q(N+M)^2+N^4 +Q(N+K)^3 + N^3+N^2 (N^2 +N+M^2+M) Q  ) + I_3 I_4 (M+(N+M)Q+2Q)^{1/2} M^2 (M^4 + M^2 (M^2 + (N+M)^2Q  ) +M^3 +Q(N+M)^3 +M^2Q (N^2 + N +M^2+N)^2 )     )$, where  $I_3$ and $I_4$ denote the maximum iteration times of the AO algorithm and the SROCR technique for solving \textbf{P7}, respectively.

% \section{Extensions to  the passive RIS scenario}

\section{Numerical results}
\label{NumericalResults}
\subsection{Simulation Setup}
In this section, numerical simulations are conducted to evaluate the performance of the proposed scheme.  The simulation setup is modeled as  a two-dimensional Cartesian coordinate as shown in Fig. \ref{Location_Illustration}.
Specifically, the PT, the SU, the RIS are deployed at (0 m, 0 m), ($x_s$, 0 m), and $(x_{r},y_{r})$, respectively. The PUs and Eves are randomly placed within  the circles  centered at $(x_{p}, 5 ~\text{m})$ and $(x_{e}, -5 ~\text{m})$ with radii $r_p$ and $r_e$, respectively. 
Similar to \cite{Hu}, we consider that all the channels are constituted by the large-scale fading and small-scale fading. The large-scale fading is modeled as $D(d)= (\zeta/(4\pi))^2 d^{-\alpha}$, where  $\zeta$ is the wavelength with a carrier frequency of 750 MHz, $\alpha$ represents the path-loss exponent, and $d$ represents the distance between two devices. Denote the path-loss exponents for the PT-RIS link, the RIS-SU link, the RIS-PUs links, and the RIS-Eves links as $\alpha_\text{PT,RIS}$, $\alpha_\text{RIS,SU}$, $\alpha_\text{RIS,PU}$, and $\alpha_\text{RIS,Eve}$, respectively. The path-loss exponents of the links unrelated with the RIS  are set at 4.
Due to the fact that the RIS can be carefully deployed, the scale-fading with respect to the RIS is considered to be the Rician fading. For example, the channel matrix from the PT to the RIS is expressed as 
$$\bm{G} = \sqrt{ \frac{\kappa_G}{\kappa_{G} + 1 }} \bm{G}^{\text{LoS}} + \sqrt{\frac{1}{\kappa_G + 1}} \bm{G}^{\text{NLoS}},$$ where $\bm{G}^{\text{LoS}}$ and $\bm{G}^{\text{NLoS}}$ denote the  line-of-sight (LoS) component and the Rayleigh fading component of $\bm{G}$, respectively, and $\kappa_G$ is the Rician factor.  $\bm{G}^{\text{LoS}}$ is modeled as $\bm{G}^{\text{LoS}}= \bm{\beta}_M(\varsigma_\text{AoA}) \bm{\beta}_N^H(\varsigma_\text{AoD})$, where $\bm{\beta}_{X}(\varsigma) = [1, e^{j\pi \text{sin}(\varsigma) }, \ldots, e^{j\pi (X-1) \text{sin}(\varsigma) }]^T$, $X\in \{M,N\}$,
$\varsigma_\text{AoA}$ and $\varsigma_\text{AoD}$ are the angle of arrival and angle of departure, respectively. Each element of  $\bm{G}^{\text{NLoS}}$ follows the CSCG distribution with zero mean and unit variance. The other channels related with the RIS can be similarly  defined.
 The channels unrelated with the RIS are modeled to be Rayleigh fading, and each element of which is also generated following $\mathcal{CN}(0,1)$. 

\begin{figure}
	\centering
	\includegraphics[width=0.7\linewidth]{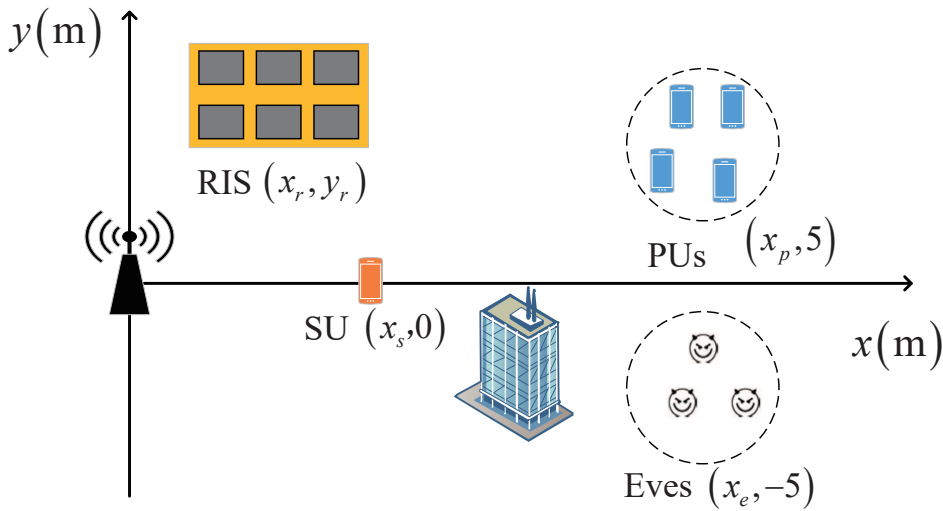}
	\caption{Simulation setup for the RIS assisted secure SR system.}
	\label{Location_Illustration}
\end{figure}
%fig3-4

\begin{figure*}
	  \begin{minipage}[t]{0.50\textwidth}
		\centering
		\includegraphics[width=1 \linewidth]{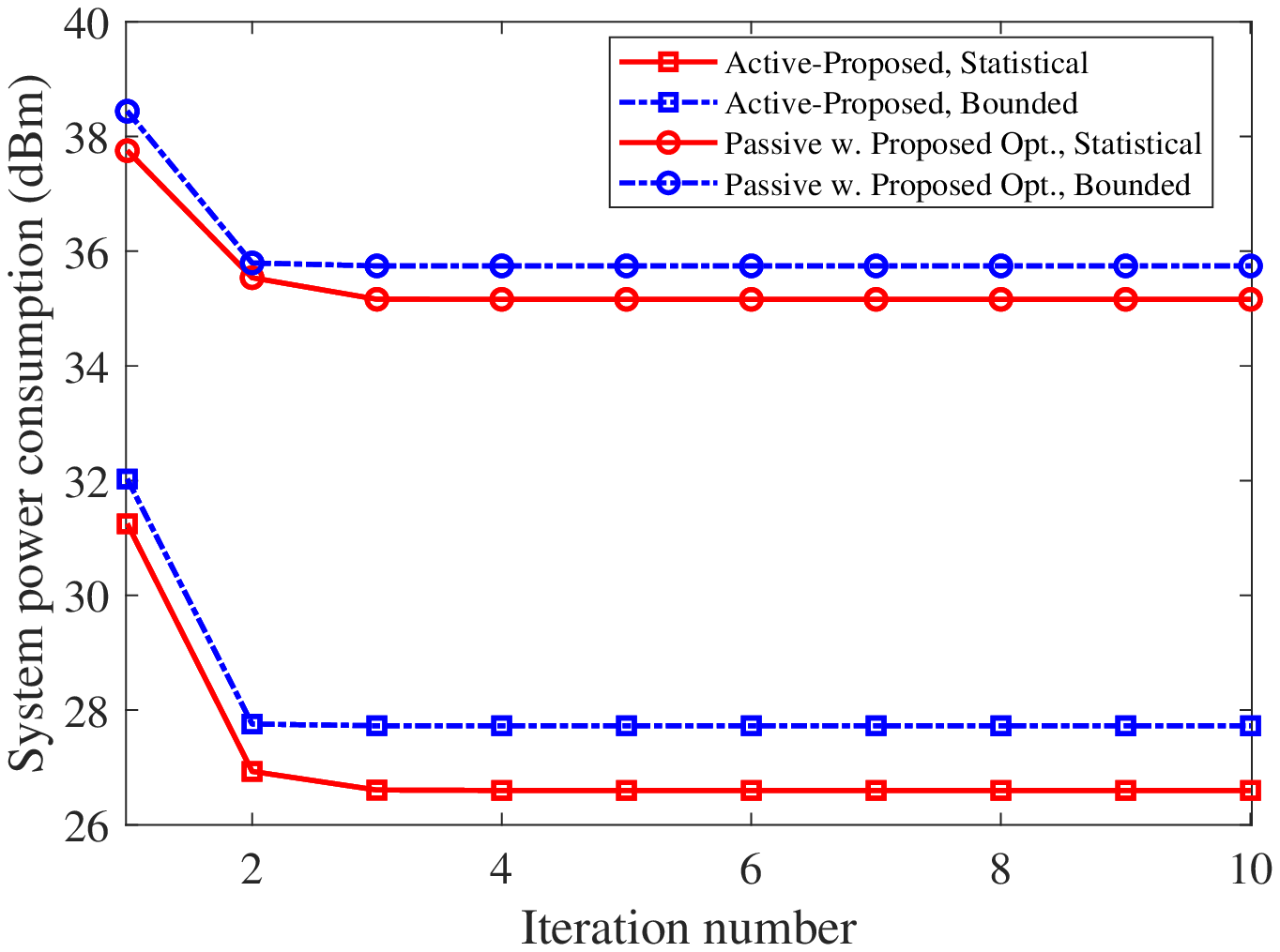}
		\caption{Convergence performance of the proposed algorithm.}
		\label{convergence}
	  \end{minipage}%
	  \begin{minipage}[t]{0.50\textwidth}
		\centering
		\includegraphics[width=1 \linewidth]{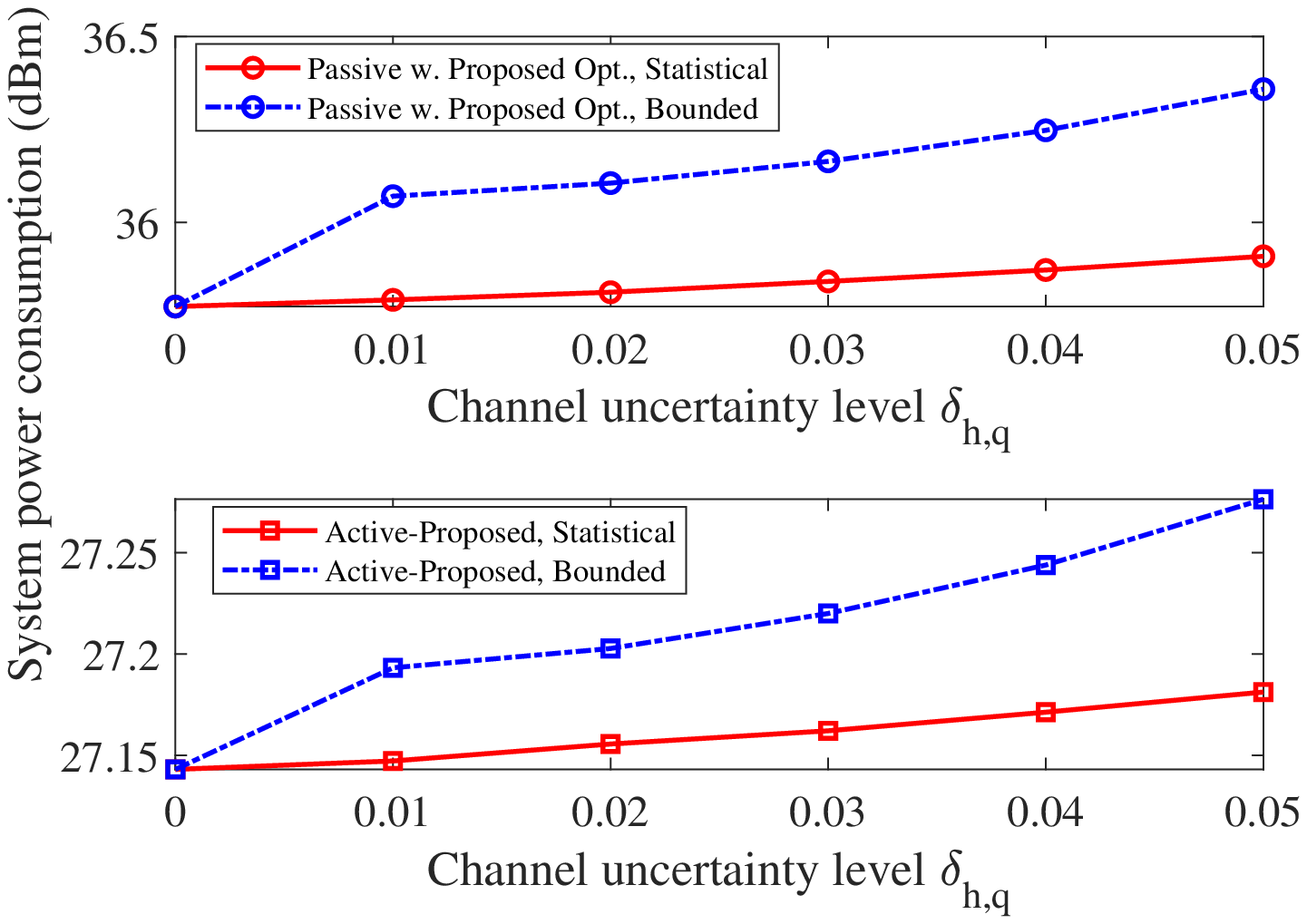}
		\caption{System power consumption versus the channel uncertainty \protect\\ level $\delta_{h,q}$.}
		\label{errorlevel}
		\end{minipage}%
\end{figure*}

According to \cite{Bernstein-Type}, for the statistical CSI error model, we model the variance of $\Delta g_{e,q}$ and $\Delta h_{e,q}$ as  $ \epsilon_{g,q}^2=\delta_{g,q}^2 {\left| \left|\hat {\bm{g}}_{e,q}\right| \right|}_2^2 $ and $ \epsilon_{h,q}^2=\delta_{h,q}^2 {\left| \left|\hat {\bm{h}}_{e,q}\right| \right|}_2^2 $, respectively, where 
 $ \delta_{g,q} \in [0,1) $ and  $ \delta_{h,q} \in [0,1) $ are used to  indicate the relative amount of CSI uncertainties of $\bm{h}_{e,q}$ and $\bm{g}_{e,q}$, respectively.  
Denote the  inverse cumulative distribution function  of the Chi-square distribution as $F_{R}^{-1}(\cdot)$, where $R$ is the degree of freedom. 
For the bounded CSI error model, the radii of the uncertainty regions of $\bm{h}_{e,q}$ and $\bm{g}_{e,q}$  are respectively expressed as  
	\begin{equation}
	\begin{aligned}
	\xi_{g,q}=\sqrt{ \frac{\epsilon_{g,q}^2}{2} F_{2M}^{-1}(1-\rho_q) }, \\
	\xi_{h,q}=\sqrt{ \frac{\epsilon_{h,q}^2}{2} F_{2N}^{-1}(1-\rho_q) }. 	
	\end{aligned}
\end{equation}
It should be noted that the CSI error model described above can guarantee the fair performance comparison between the robust beamforming designs for the bounded CSI error model and the statistical CSI error model \cite{Bernstein-Type}.

 Unless otherwise specified, the other parameters are  given as follows: $x_s = 20$ m, $x_r =  16$ m, $y_r = 2$ m, $x_e = 80$ m, $x_p = 80$ m, $r_p = r_e =5$ m, $N = 6$, $M=16$,  $K =2$, $Q=2$, $L=60$, $\sigma_{p,k}^2 = \sigma_{s}^2 = \sigma_{e,q}^2 = \sigma_{I}^2= -90$ dBm, $\gamma_s = 23$ dB, $\gamma_e = 15$ dB,
 $\rho_q = 0.05$,  $ \delta_{g,q}=0.01 $, $\delta_{h,q}=0.02 $, $\varepsilon = 0.01$, $\alpha_\text{RIS,SU} = 3.2$, $\alpha_\text{RIS,Eve} = 4$, $\eta^2 = 10$,  $P_\text{SW} = -10$ dBm,  and $P_\text{DC} = -5$ dBm \cite{PanActive}.

\subsection{Benchmark Schemes}
\label{SectionBechmark}
{For performance comparison, the passive RIS scheme with proposed optimization (abbreviated as Passive w. Proposed Opt.) can be used as a benchmark.} Specifically, for the passive RIS scenario,  by setting $\sigma_I^2 = 0$ and $\eta=1$, the corresponding SNRs of decoding  $s(l)$ and $c$ at different nodes  can be straightforward to obtain  based on \eqref{SNRPU}, \eqref{SNREQ}, \eqref{SNRSS}, and \eqref{SNRSC}, respectively. According to \cite{PanActive}, the total system power consumption for the passive RIS scenario is expressed as 
\begin{align}
P_\text{tol}^\text{passive} = P_\text{PT}+ M P_\text{SW}.
\end{align} 
Similar to Section \ref{BoundedCSI} and Section \ref{StatisticalCSI},  the system power consumption minimization problems  for the passive RIS scenario in  the bounded CSI error model and the statistical CSI error model can also be solved by our proposed algorithm.  {Similarly, the hybrid RIS scheme with proposed optimization (abbreviated as Hybrid w. Proposed Opt.) is also used as a benchmark, in which the number of active elements is $M/2$.

Moreover, the schemes with random reflection beamforming and the scheme using a backscatter device (BD) with one antenna (abbreviated as BD scheme) are considered as the benchmarks.} Moreover, the non-robust beamforming schemes with optimization based on the estimated CSI related with Eves (i.e., $\hat{\bm{h}}_{e,q}$ and $\hat{\bm{g}}_{e,q}$) are employed for performance comparison.

 % \subsection{Benchmark schemes}
 % For performance comparison, the scheme  with random reflection beamforming is considered as the benchmark scheme.
%  \begin{itemize}
% \item{Passive RIS scheme:} For this scheme,  by setting $\sigma_I^2 = 0$ and $\eta=1$, the corresponding SNRs of decoding  $s(l)$ and $c$ at different nodes  can be straightforward to obtain  based on \eqref{SNRPU}, \eqref{SNREQ}, \eqref{SNRSS}, and \eqref{SNRSC}, respectively. In addition, the total system power consumption is expressed as 
% $P_\text{tol}^\text{passive} = P_\text{PT}+ M P_\text{SW}$ \cite{PanThree}. The minimization problems of the total system power consumption for the passive RIS scheme in both the bounded CSI error model and the statistical CSI error model can also be solved by our proposed algorithm.
% \item{Random reflection beamforming:} For this scheme, the reflection beamforming for  random, but the active beamforming is optimized.
% % \footnote{It is worth noting that there have been no works investigating the secure transmission for RIS assisted SR systems under the channel uncertainty scenario in the literature.}. 
% \item{XX}
%  \end{itemize}

 \subsection{Performance Evaluation}

%  \begin{figure}
% 	\centering
% 	\includegraphics[width=0.7 \linewidth]{convergence}
% 	\caption{Convergence performance of the proposed algorithm.}
% 	\label{convergence}
% \end{figure}

In Fig. \ref{convergence}, we first investigate the convergence performance of the proposed algorithm with    $\gamma_c = 65$ dB, $\alpha_\text{PT,RIS}=2.4$ and $\alpha_\text{RIS,PU}=2.4$. As observed in Fig. \ref{convergence}, our proposed algorithm can converge to  the predefined threshold (i.e., $ \varepsilon$) after only several iterations, which indicates that the efficiency of our proposed algorithm is satisfying. 

% \begin{figure}
% 	\centering
% 	\includegraphics[width=0.7\linewidth]{error_level.eps}
% 	\caption{System power consumption versus the channel uncertainty level $\delta_{h,q}$.}
% 	\label{errorlevel}
% \end{figure}

In Fig. \ref{errorlevel}, we investigate the  system power consumption versus the channel uncertainty level $\delta_{h,q}$ with  $\gamma_{c}=65$ dB, $\alpha_\text{PT,RIS}=2.4$ and $\alpha_\text{RIS,PU}=2.4$. {As observed in Fig. \ref{errorlevel},  the proposed active RIS scheme consumes less power than that of the passive RIS scheme with proposed optimization. For example, when $\delta_{h,q}=0.05$, compared to the passive RIS scheme, the proposed active RIS can save  25.0\% and 24.3\%  system power for the  bounded and statistical CSI error models, respectively.}
 This observation can be explained in two aspects. On the one hand,  the active RIS can amplify the incident signal to significantly  enhance the communication efficiency of both primary and secondary transmissions. On the other hand, the confidential information intercepted by the Eves is destructed severely with the help of the active RIS. {In addition, the  thermal noise generated at the active  RIS can be considered as the artificial noise  to interfere the Eves by reducing the SNR for decoding $s(l)$.} Thus, less power is required at the PT for the active RIS scheme to satisfy all the constraints. It is noted that compared to the passive RIS, the additional power consumption of using the  active RIS is generally limited. 
% It is observed that  as  $\delta_{h,q}$ becomes large, the  system power consumption for the proposed active RIS scheme almost keeps unchanged, which confirms its robustness to the channel uncertainty.

 % Moreover, it is also seen that the optimization of reflection beamforming can enhance the system energy efficiency. For example, when $\delta_{h,q}=0.05$, compared to the active RIS scheme with  random reflection beamforming,  the proposed active RIS scheme with  optimized reflection beamforming can save 2.49 dBm system power consumption for the statical CSI error model.

%fig5-6-7-8
 \begin{figure*}
  \begin{minipage}[t]{0.50\textwidth}
    \centering
    \includegraphics[width=1 \linewidth]{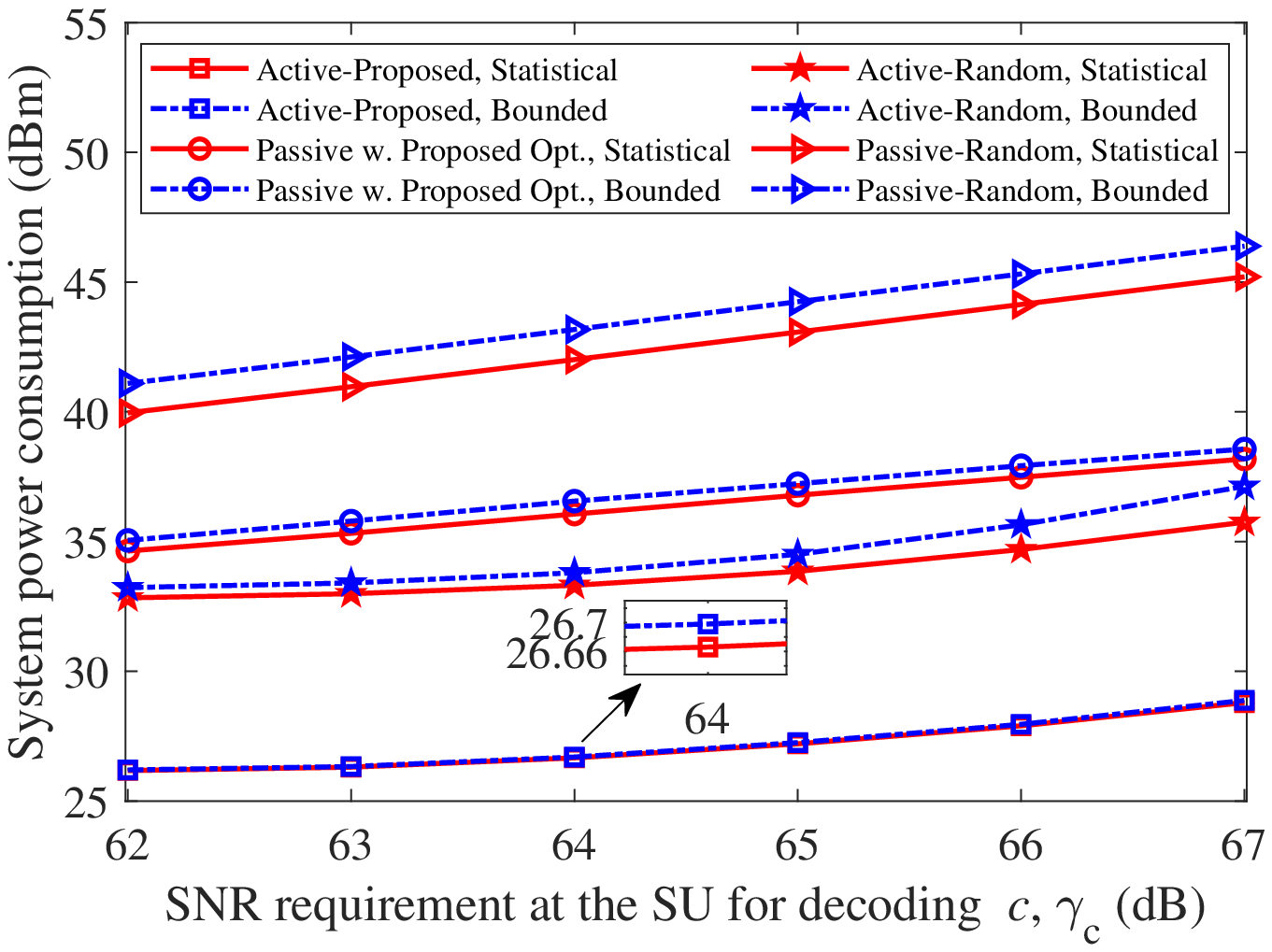}
    \caption{System power consumption versus the SNR requirement \protect\\for decoding $c$  with different schemes.}
    \label{gamma_c1}
  \end{minipage}%
  \begin{minipage}[t]{0.50\linewidth}
    \centering
    \includegraphics[width=1 \linewidth]{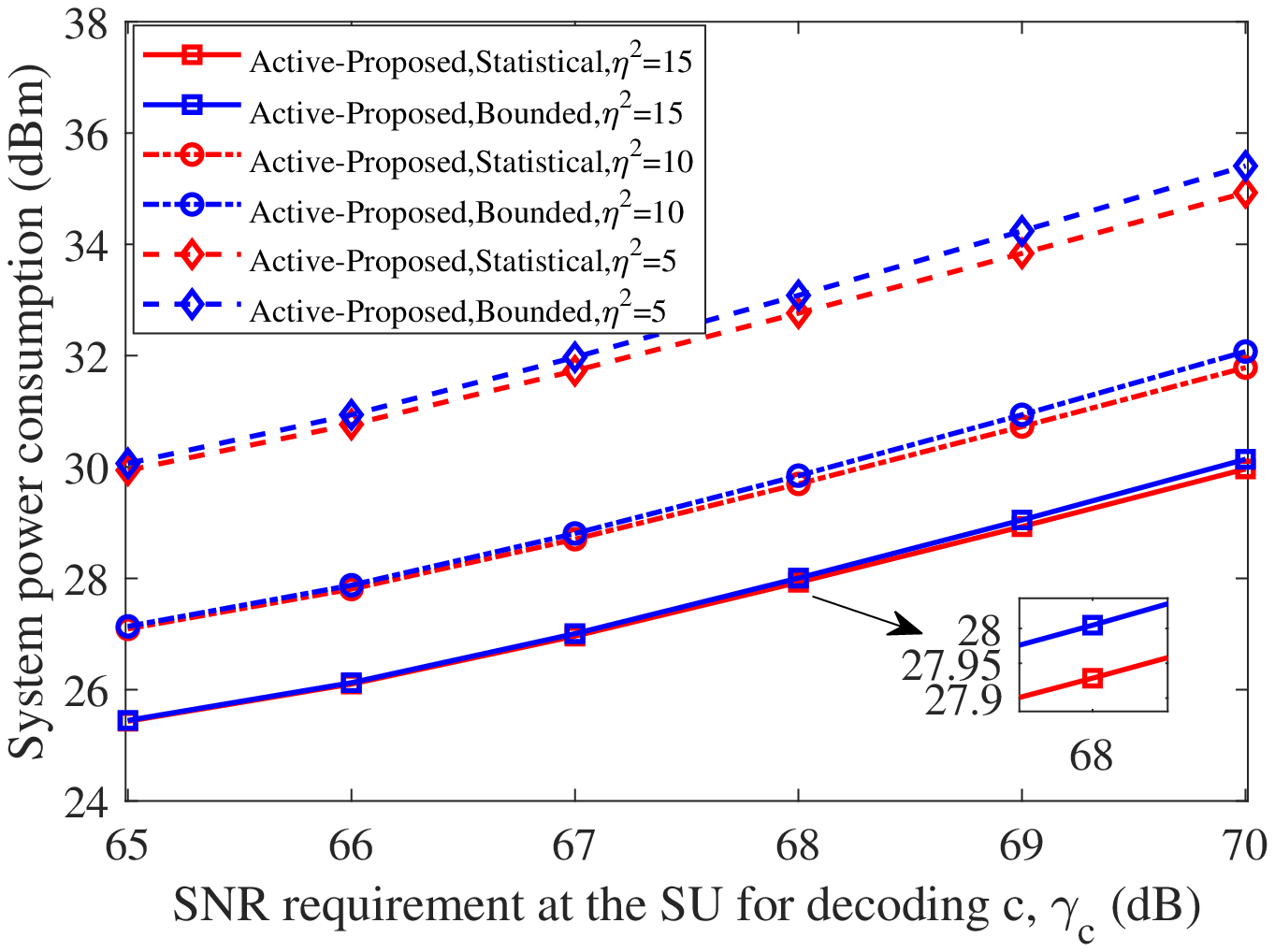}
    \caption{System power consumption versus the SNR  requirement \protect\\for  decoding $c$ with different amplification factors.}
    \label{gamma_c2}
  \end{minipage}
\end{figure*}
 \begin{figure*}
	\begin{minipage}[t]{0.49\textwidth}
		\centering
		\includegraphics[width=1 \linewidth]{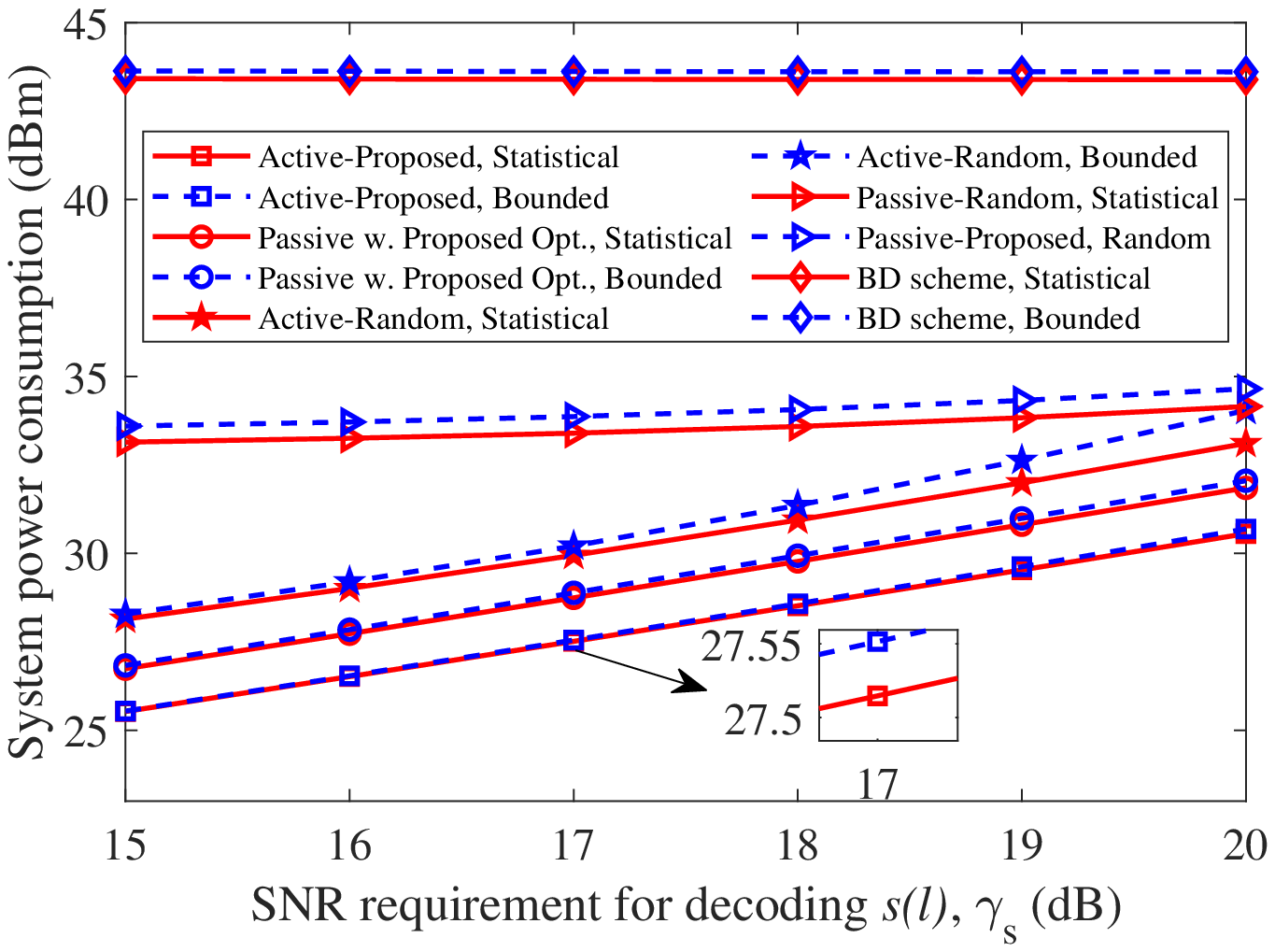}
		\caption{System power consumption versus the SNR requirement \protect\\for  decoding $s(l)$.}
		\label{gamma_s}
	\end{minipage}%
	\begin{minipage}[t]{0.49\linewidth}
		\centering
		\includegraphics[width=1 \linewidth]{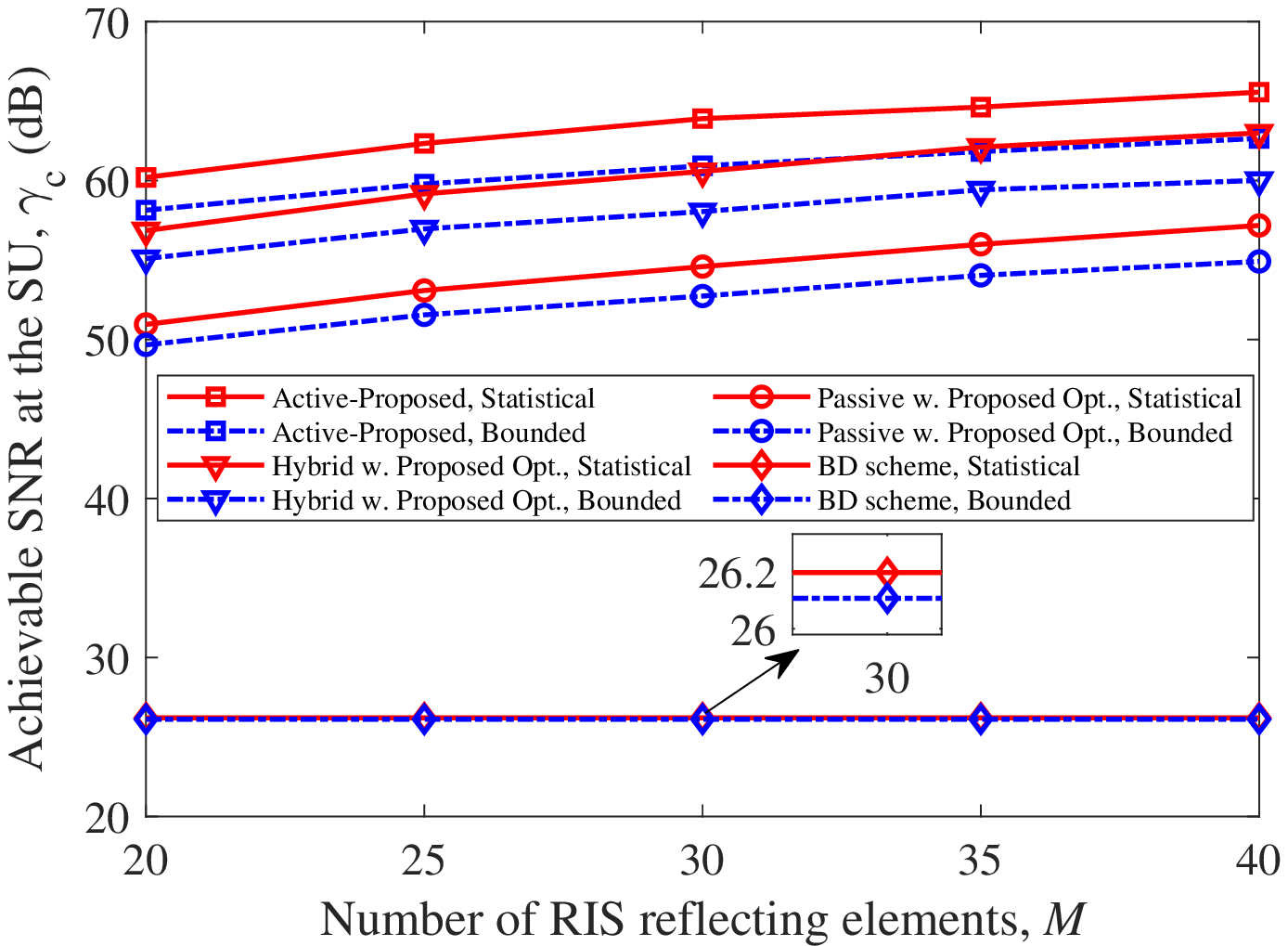}
		\caption{Achievable SNR for decoding $c$ versus the number of reflecting elements.}
		\label{Knumber}
	\end{minipage}
\end{figure*}

% \begin{figure}
% 	\centering
% 	\includegraphics[width=0.7\linewidth]{gamma_c1.eps}
% 	\caption{System Power consumption versus the SNR  for decoding $c$ with different schemes.}
% 		\label{gamma_c1}
% \end{figure}

In Fig. \ref{gamma_c1}, we study the impact of the SNR requirement at the SU for decoding the secondary signal (i.e., $c$) on the system power consumption with $\alpha_\text{PT,RIS}=2.4$ and $\alpha_\text{RIS,PU}=2.4$. {First of all, we observe that the system power consumption functions of all schemes are  increasing functions with respect to the SNR requirement at the SU for decoding $c$. This is because as the SNR requirement at the SU for decoding $c$ improves, more transmit power at the PT is needed to satisfy the SNR constraints.} However, the  system power consumption for the proposed  active RIS scheme  in the bounded and statistical CSI error models is the least. This is due to the fact that the deployment of the active RIS can significantly improve the secondary transmission efficiency. Moreover, it is seen that the optimization of reflection beamforming can enhance the system energy efficiency. For example, when $\gamma_c = 64$ dB, compared to the active RIS scheme with random reflection beamforming, the proposed active RIS scheme with optimized reflection beamforming can save 6.65 dBm system power for the statistical CSI error model. 

%fig9-10
 \begin{figure*}
	\begin{minipage}[t]{0.5\textwidth}
		\centering
		\includegraphics[width=1 \linewidth]{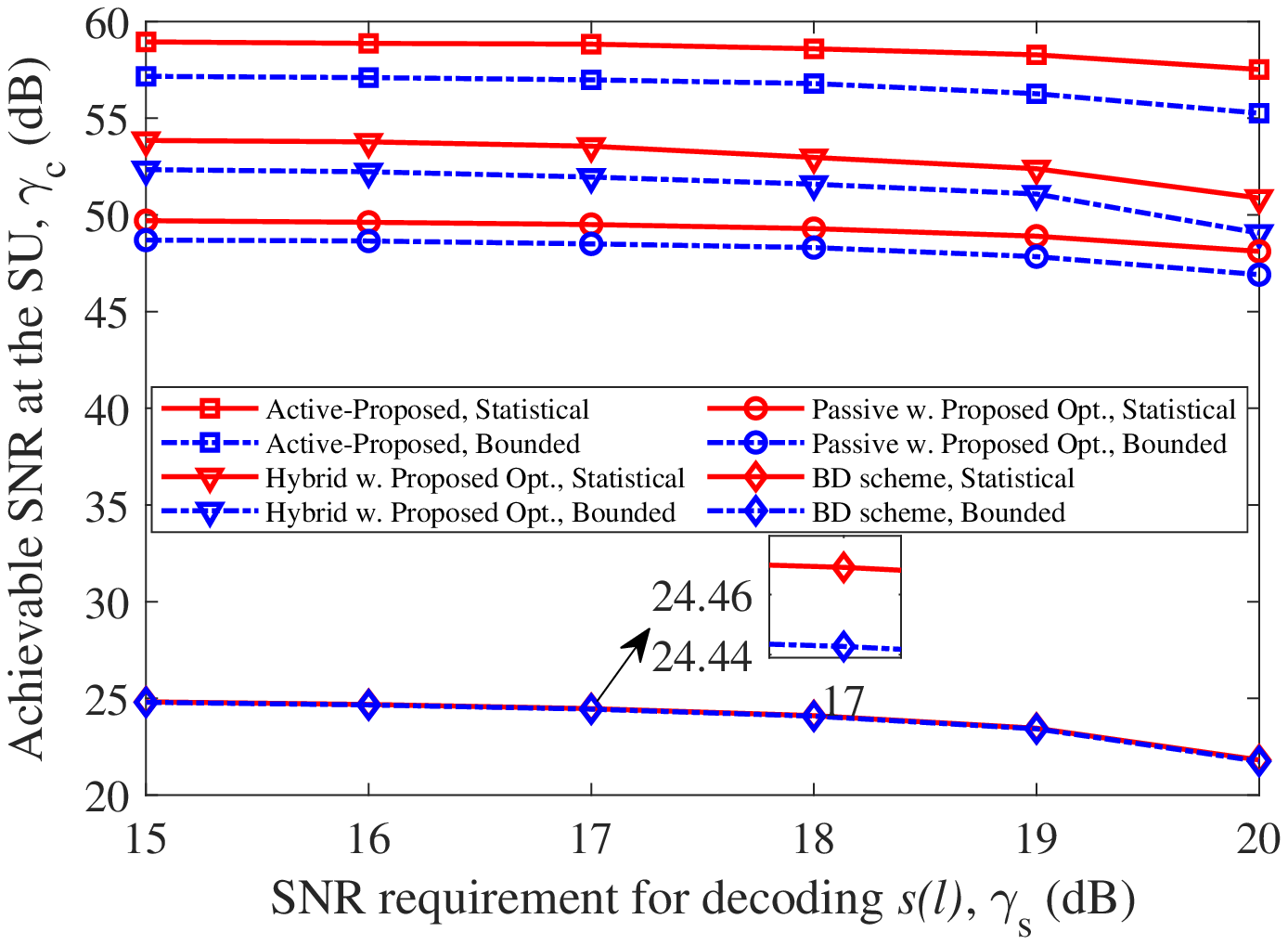}
		\caption{Achievable SNR for decoding $c$ versus  the SNR \protect\\requirement for  decoding $s(l)$.}
		\label{gammaCvsgammaS}
	\end{minipage}%
	\begin{minipage}[t]{0.5\linewidth}
		\centering
		\includegraphics[width=1 \linewidth]{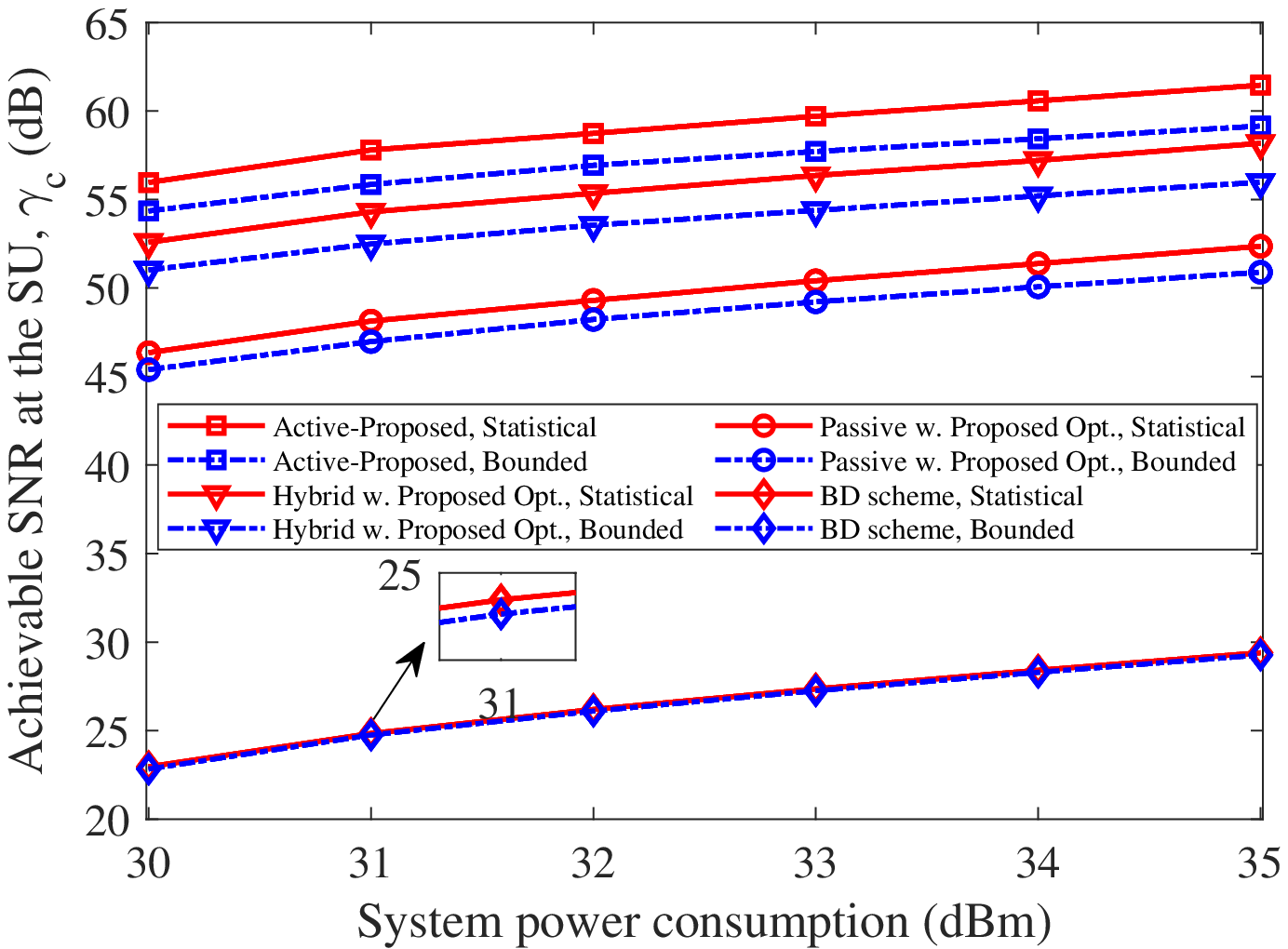}
		\caption{Achievable SNR for decoding $c$ versus  the system \protect\\  power consumption.}
		\label{gammaCvsPower}
	\end{minipage}
\end{figure*}

% \begin{figure}
% 	\centering
% 	\includegraphics[width=0.7\linewidth]{gamma_c2.eps}
% 	\caption{System Power consumption versus the SNR  for decoding $c$ with different amplification factors.}
% 		\label{gamma_c2}
% \end{figure}

{Fig. \ref{gamma_c2} shows the system power consumption for the proposed active RIS scheme versus the SNR requirement at the SU for decoding $c$ with   different amplification factors with $\alpha_\text{PT,RIS}=2.4$ and $\alpha_\text{RIS,PU}=2.4$. As observed in Fig. \ref{gamma_c2}, the schemes with 
a larger amplification factor generally consume less power because the secondary transmission efficiency is positively correlated with the amplification factor. With the same amplification factor,   the scheme with the statistical CSI error model always consumes less power than that with the bounded CSI error model. This is because for the bounded CSI error model, more power is needed to meet the SNR requirement at the Eves for the channel generation with the worst-case error. Again, we observe that the increase of $\gamma_c$ results in more system power consumption for all schemes, which shows the impact of the SNR  requirement at the SU on the system performance.}

% \begin{figure}
% 	\centering
% 	\includegraphics[width=0.7\linewidth]{gamma_s.eps}
% 	\caption{System Power consumption versus the SNR for decoding $s(l)$.}
% 		\label{gamma_s}
% \end{figure}

Fig. \ref{gamma_s} depicts the effect of the SNR requirement for decoding the primary signal $s(l)$ on the system power consumption with $\gamma_c = 40$ dB, $\alpha_\text{PT,RIS}=3.2$, and $\alpha_\text{RIS,PU} = 4$. As the SNR requirement for decoding $s(l)$ increases, the required power at the PT for satisfying all constraints improves. It is seen that 
 the proposed active RIS scheme with the statistical CSI error model consumes the least power due to the transmission efficiency enhancement provided by using the active RIS and the accurately robust design achieved by the statistical CSI error model.  Compared to the proposed schemes,  the gaps between the bounded and statistical CSI error models for  the schemes with random reflection beamforming  are larger. This is due to the fact that the performance enhancement by the random reflection beamforming is very limited, which results in more consumed power to address the worst-case error realization following  the bounded CSI error model. {Furthermore, we observe that the system power consumption of the BD scheme is the largest since the performance enhancement provided by the BD with one antenna is very limited. This observation confirms the superiority of using the RIS for performance improvement in SR systems.}

% \begin{figure}

% 		\centering
% 		\includegraphics[width=0.7\linewidth]{Knumber}
% 		\caption{System power consumption versus the number of reflecting elements.}
% 		\label{Knumber}
% \end{figure}

{The impact of the number of RIS reflecting elements on the  achievable SNR for decoding the secondary signal at the SU with   $P_\text{tol}^\text{active} = 32$ dBm, $\alpha_\text{PT,RIS}=3.2$, $\alpha_\text{RIS,PU}=4$ and $\gamma_s=20$ dB is shown in Fig. \ref{Knumber}.  As observed in Fig. \ref{Knumber}, increasing the number of RIS reflecting elements can reduce the system power consumption for the RIS enabled schemes. The reason is that  by deploying more reflecting elements at the RIS, more transmission links can be provided for enhancing the communication efficiency of  the  secondary transmission.  Moreover, compared to the hybrid RIS scheme and the passive RIS scheme with proposed optimization, the proposed active RIS scheme achieves a higher SNR for decoding the secondary signal at the SU. For example, the achievable SNRs of the proposed active RIS scheme with $M=20$ for the  bounded and statistical CSI error models are 58.3 dB and 60.2 dB, respectively, which are even larger than those achieved by the passive RIS scheme with $M=40$.
It is worth noting that the SNR achieved by the BD scheme for decoding the secondary signal at the SU is constant since the BD has only one antenna.}

\begin{figure*}
	\begin{minipage}[t]{0.5\linewidth}
		\centering
		\includegraphics[width=1 \linewidth]{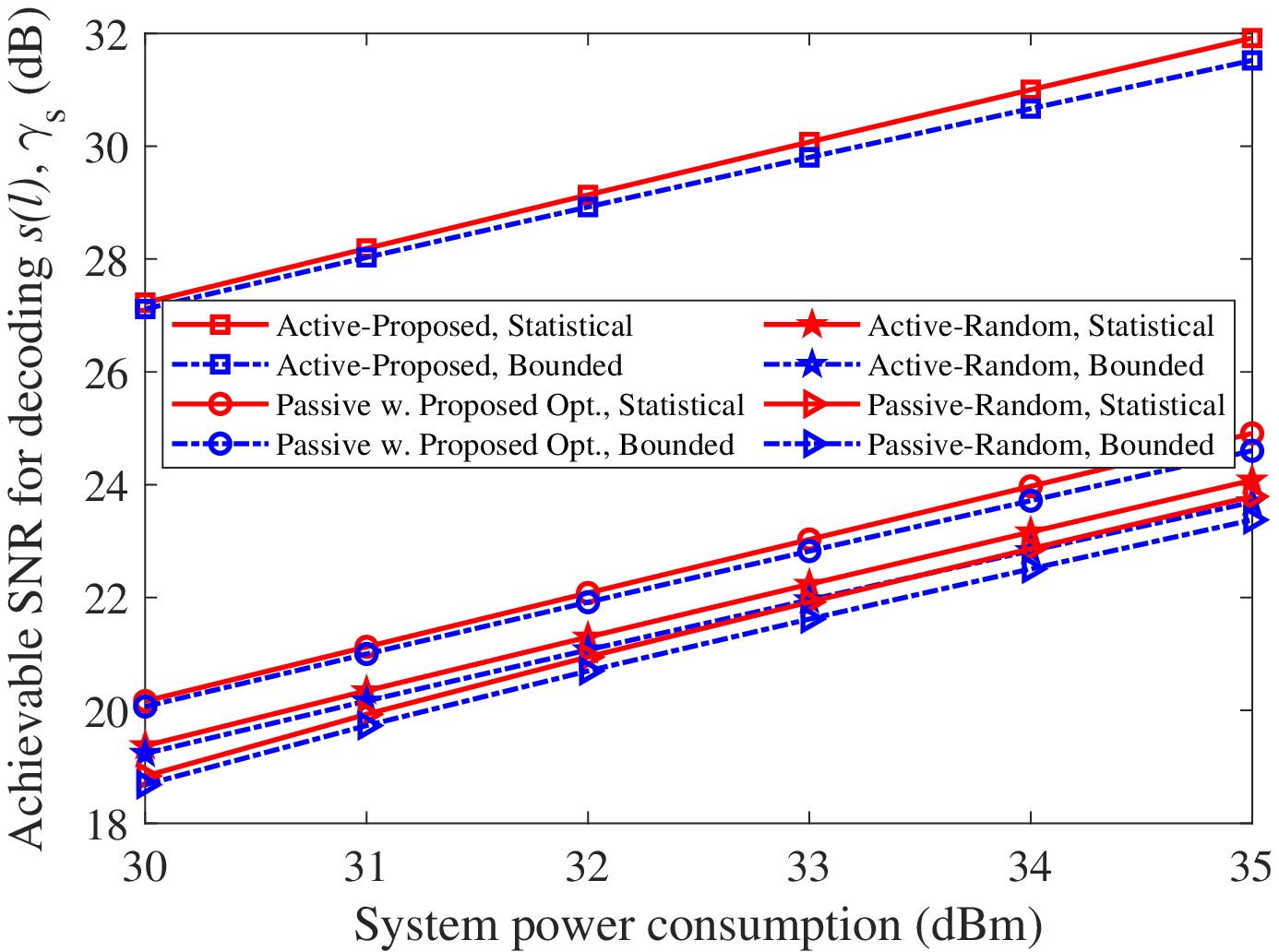}
		\caption{Achievable SNR for decoding $s(l)$ versus the  system \protect \\power consumption.}
		\label{gammaSvsPower}
	\end{minipage}
	\begin{minipage}[t]{0.5\textwidth}
		\centering
		\includegraphics[width=1 \linewidth]{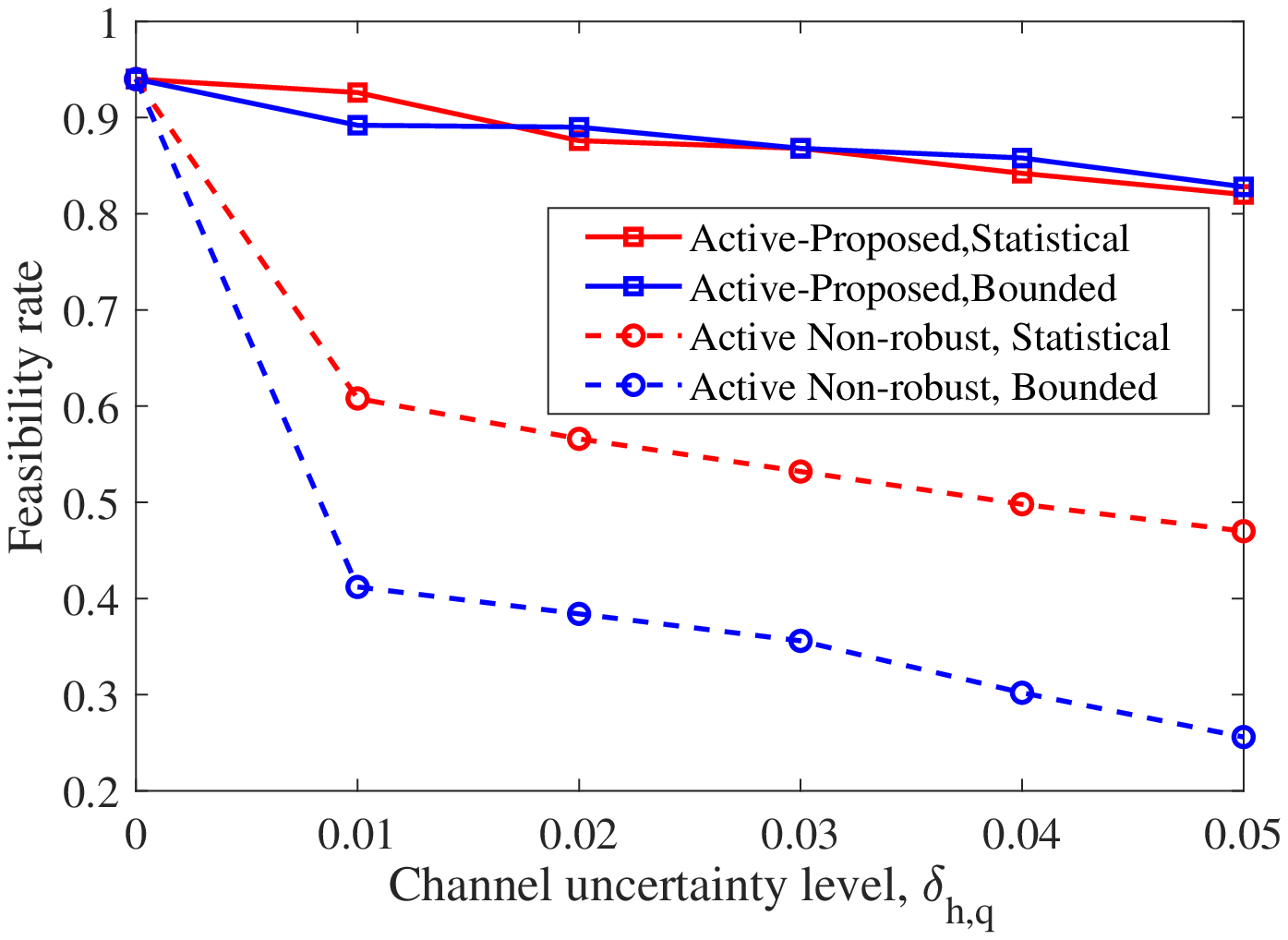}
		\caption{Feasibility rate versus the channel uncertainty level $\delta_{h,q}$.}
		\label{feasibility}
	\end{minipage}%
\end{figure*}

In Fig. \ref{gammaCvsgammaS}, we evaluate the achievable SNR for decoding the secondary signal at the SU versus the SNR requirement for decoding the primary signal with $P_\text{tol}^\text{active} = 32$ dBm, $\alpha_\text{PT,RIS}=3.2$, and $\alpha_\text{RIS,PU} = 4$. It is found that for each CSI error model, the proposed active RIS scheme can achieve the largest SNR for decoding  $c$, which again indicates that proposed active RIS scheme can significantly improve the secondary transmission efficiency. By improving the SNR requirement for decoding $s(l)$, the achievable SNRs for decoding $c$ of all the schemes reduce. This is due to the fact that the transmit beamforming and reflection beamforming should be configured towards  the PUs to satisfy their SNR requirements. {It is observed that the SNR at the SU achieved by the hybrid RIS scheme is larger than that achieved by the passive RIS scheme but smaller than that achieved by the proposed active RIS scheme, which implies that the number of  active elements is an important factor affecting the performance enhancement. Furthermore, we find that the SNR at  the SU achieved by the BD scheme is the smallest.}

 Fig. \ref{gammaCvsPower} shows the  achievable SNR for decoding the secondary signal at the SU versus the system power consumption with $\gamma_s = 18$ dB, $\alpha_\text{PT,RIS}=3.2$, and $\alpha_\text{RIS,PU}=4$. By increasing the system power consumption, the achievable SNRs for all the schemes improve. The reason is that as the system power consumption increases, the PT can transmit primary signal with a higher power, which thus improves the received signal power at the PUs and SU. 
{As observed in Fig. \ref{gammaCvsPower}, compared to the BD scheme in the bounded and statistical CSI error models, the proposed active RIS scheme can achieve up to 138.03 \% and 143.76\% performance gain in terms of the SNR for decoding $c$, respectively, which again confirms the superiority of our proposed active RIS scheme. In addition, the performance comparison between the proposed active RIS scheme and the hybrid RIS scheme is consistent with the observation in Fig. \ref{gammaCvsgammaS}.}

 {Fig. \ref{gammaSvsPower} shows the  achievable SNR for decoding the primary signal versus the system power consumption with $\gamma_c = 45$ dB, $\alpha_\text{PT,RIS}=3.2$, and $\alpha_\text{RIS,PU}=4$. As observed in Fig. \ref{gammaSvsPower}, increasing the system power consumption improves the achievable SNR for decoding $s(l)$ since more power can be allocated to the PT for confidential information delivery. Similarly,  the proposed active RIS scheme  achieves the best performance. Compared to the passive RIS scheme with proposed optimization, the proposed active RIS scheme can achieve up to  35.11\% and 35.02\% performance gain for the  bounded and statistical CSI error models, respectively. Again, the performance gaps between the proposed active RIS scheme with proposed optimization and random reflection beamforming verify the importance of optimally adjusting the amplitude reflection coefficients and phase shifts at the active RIS.
 }

The impact of the channel uncertainty level $\delta_{h,q}$ on the feasibility rate of \textbf{P1} and \textbf{P7}  is depicted in Fig. \ref{feasibility}. As observed in Fig. \ref{feasibility}, the proposed active RIS scheme with robust beamforming design can achieve satisfying feasibility rates, e.g., the feasibility rates for the bounded and statistical CSI error models  can exceed 0.82 even if $\delta_{h,q} =0.05$. However, as $\delta_{h,q}$ increases, the non-robust beamforming design has a high probability of failing to satisfy the constraints. For example, when $\delta_{h,q} =0.05$, the feasibility rate of the non-robust beamforming design for the bounded CSI error model is only 0.256. This  observation indicates that the proposed robust beamforming design works well for the   imperfect CSI scenarios.

\section{Conclusions}
\label{Conculsions}
This paper has investigated the secure transmission for the active RIS enabled SR system against  the existence of Eves. The active RIS has been adopted to achieve two purposes, i.e., the first one is to assist  the confidential  information transmission from the PT to multiple PUs, and the second one is  to deliver its own information by modulating and reflecting the primary signal. Taking into account the fact that the Eves are usually unknown, we have considered the bounded CSI error model and statistical CSI error model to capture the uncertainty of acquiring the CSI related with Eves. For  each model, we have formulated the system power consumption minimization problem by  designing the robust transmit beamforming and reflection beamforming. Specifically, for the bounded CSI error model, we have considered the worst-case SNR constraints at the Eves and implemented the S-Procedure to address the characteristic of having semi-infinite inequalities. While,   the SNR outage probability constraints have been considered for the statistical CSI error model and properly transformed  into tractable forms  by using the  Bernstein-Type Inequality.
The AO algorithm with SDR and SROCR techniques has been proposed to solve the formulated problems efficiently. Moreover, the tightness of applying the SDR has been proved, which results in a highly accurate solution. Extensive numerical results have been provided to confirm the convergence of  the proposed algorithm, and they also demonstrate that the proposed schemes always consume less power compared to  the benchmark schemes.

\appendix

\section{Proof of Theorem \ref{RankOne}}
The Lagrangian function of $\textbf{P3}$ with respect to $\bm W$ is given by
% \begin{equation}
% 	\begin{aligned}\label{lagrangian}
		$\mathcal{L}(\bm{W})= \text{Tr}( \bm{W}) + \text{Tr} (\bm{\Upsilon} \bm{W}) -\sum_{k=1}^K \lambda_k (\text{Tr}({\bm{H_}{p,k}}\bm{W}) +\text{Tr}({\bm{G}_{p,k}}\bm{W})) -\mu(\text{Tr}({\bm{H}_s}\bm{W})+ \text{Tr}({\bm{G}_s}\bm{W}))
				-\xi L  \text{Tr}({\bm{G}}_{s}\bm{W}) 
				-\text{Tr}(\bm{\Omega} \bm W)- \sum_{q=1}^Q \text{Tr}(\bm{M}_q \hat{\bm{W}_q})+\eta,$
where  $\lambda_k \ge 0$, $\mu \ge 0 $, $\xi\ge 0$, $\bm{M}_q \succeq 0$, and $\bm{\Omega} \succeq 0$  are the Lagrangian multipliers with respect to the constraints C6, C7, C8, C9, and C11, respectively, and $\eta$ is the term unrelated with $\bm W$.  
The Karush–Kuhn–Tucker (KKT) conditions of \textbf{P3} are expressed as 
\begin{align}
\label{Partial}
	&\frac{\partial{\mathcal{L}}}{\partial{\bm{W}}} = \bm{{I}}_{N } - \bm{B} - \bm{\Omega}^* =\bm 0, \\
	\label{RankOne}
	& \text{Tr} (\bm{\Omega}^* \bm{W}^*) = 0,
	% & \text{Tr} (\bm{M}^* \hat{\bm{W}})=0,
\end{align}
where $ \bm B= \sum_{k=1}^K \lambda_k^*(\bm{H}_{p,k}+\bm{G}_{p,k})+ \mu^*(\bm{H}_{s}+\bm{G}_{s})+\xi^*L  \bm{G}_{s}+ \sum_{q=1}^Q \hat{\bm{M}}_q - \bm{\Upsilon}$, $\lambda_k^*$, $\mu^*$, $\xi^*$ and $\bm{\Omega}^*$ are the optimal Lagrangian multipliers, and $\hat{\bm{M}_q} \in \mathcal{C}^{N\times N}$ denotes the value of $\frac{\partial{\text{Tr}(\bm{M}_q \hat{\bm{W}_q})}} {\partial{\bm{W}}}$ \cite{Matrix}.

From \eqref{RankOne}, we observe that  $\bm{W}^*$ lies in the null space of $\bm{\Omega}^*$. Thus, we analyze the rank of $\bm{W}^*$ by considering 
the structure of $ \bm{I}_N-\bm{B} $ (i.e., $\bm{\Omega}^*$)  according to \eqref{Partial}. Denote the largest eigenvalue of $\bm{B}$ as $\vartheta_\text{max}$. If
 $\vartheta_\text{max}>1$, the smallest eigenvalue of $\bm{{I}}_{N } - \bm{B}$ is smaller than zero, which contradicts with that  $ \bm{\Omega^*}$ should be a positive semi-definite matrix. If $\vartheta_\text{max} <1$, $\bm W^*=\bm{0}$ due to the fact that $ \bm{\Omega^*} $ is a positive semi-definite matrix with full rank, which is  obviously not the optimal solution to \textbf{P3}. If $\vartheta_\text{max}=1$, the null space of $\bm{\Omega^*}$ is spanned by $ \bm{\omega}_\text{max} \in \mathcal{C}^{N \times 1} $, which is a unit-norm eigenvector of $\bm B$ with respect to $\vartheta_\text{max}$. Hence,  there exists a rank-one transmit beamforming vector expressed by $\bm W =\alpha \bm{\omega}_{max} \bm{\omega}_{max}^H $, where $\alpha >0$ is a scaling factor.

This completes the proof of Theorem \ref{TheoremRankOne}.

\end{document}